\documentclass[11pt,a4paper]{article}
\usepackage[a4paper, margin=2.5cm]{geometry}
\usepackage{authblk} 
\usepackage{enumitem} 
\usepackage{amsmath}
\usepackage{amssymb}
\usepackage{amsthm}
\usepackage{currfile}
\usepackage{IEEEtrantools}
\usepackage{graphicx}
\usepackage[hidelinks]{hyperref}
\usepackage{xcolor}
\definecolor{llgray}{gray}{0.95}
\usepackage{listings}
\usepackage{longtable} 
\usepackage{import} 
\usepackage{pdflscape} 
\usepackage{afterpage} 

\makeatletter
\g@addto@macro\bfseries{\boldmath}
\makeatother

\newtheorem{theorem}{Theorem}
\newtheorem{corollary}{Corollary}
\newtheorem{lemma}{Lemma}
\newtheorem{proposition}{Proposition}
\theoremstyle{definition}
\newtheorem{definition}{Definition}[section]
\theoremstyle{definition}

\theoremstyle{definition}

\newtheorem{remark}{Remark}
\newtheorem{criterion}{Criterion}
\newtheorem{examples}{Examples}[section]
\theoremstyle{definition}

\newcommand{\ol}[1]{\overline{#1}}

\newcommand{\gr}[1]{\textcolor{gray}{#1}}

\title{On Low Rank Fusion Rings}
\author[1]{G. Vercleyen\footnote{\url{gert.vercleyen@protonmail.com}}}
\author[1,2]{J. K. Slingerland\footnote{\url{joost@thphys.nuim.ie}}}
\affil[1]{\small Department of Theoretical Physics, Maynooth University, Ireland.}
\affil[2]{\small Dublin Institute for Advanced Studies, School of Theoretical Physics, 10 Burlington Rd, Dublin, Ireland.}

\begin{document}
\maketitle
\begin{abstract}
  We present a method to generate all fusion rings of a specific rank and multiplicity. This method was used to generate exhaustive lists of fusion rings up to order 9 for several multiplicities. We introduce a class of non-commutative fusion rings based on a group with transitive action on a set. This generalizes the Tambara-Yamagami (TY) and Haagerup-Izumi (HI) fusion rings. We give an example of such a ring which is categorifiable and is not of TY or HI type. The structure of non-commutative fusion rings with a subgroup is reviewed, and the one-and two-particle extensions of groups are classified. A website containing data on fusion rings is introduced, and an overview of a Wolfram Language package for working with these rings is given.
\end{abstract}
\section{Fusion Rings}
Fusion rings, and the closely related fusion algebras, are structures that arise in various situations in mathematics and physics where a combination of two ‘simple objects’ returns an integer
amount of simple objects. In a physical setting, one typically encounters these structures when studying decompositions of tensor products of finite-dimensional representations, the addition of spin being a standard example.
Another prominent example is the study of anyon models, which have profound applications in topological quantum computing \cite{Freedman98,KITAEV20032,freedman2003topological,nayak2008non} and topological phases of matter.
An anyon model has a finite set $I$ of superselection sector labels called topological-or anyonic charges. These can be thought of as (generalised) conserved charges that distinguish between excitations of a physical system with different topological interactions. The topological charges can be combined (fused) according to a specific set of fusion rules,
\begin{equation}   \label{eq:fusionrules}
   a \times b = \sum_{c \in I} N_{ab}^c\ c,
\end{equation}
where
\begin{enumerate}[label=(\alph*)]
   \item the combination of charges is an associative operation, i.e. $(a \times b) \times c = a \times (b \times c)$,
   \item the fusion multiplicities $N_{ab}^c$ are natural numbers that indicate the number of ways the charges $a$ and $b$ can be combined to produce the charge $c$
   \item there exists a `vacuum charge' $1 \in I$ for which $N_{a1}^c = N_{1a}^c = \delta_{ac}$,
   \item every charge $a$ has a unique conjugate (or dual) charge $\ol{a}$ such that $N_{ab}^1 = \delta_{a\ol{b}} = \delta_{\ol{a}b}$ (from which it follows that $\ol{1} = 1$ and $\ol{\ol{a}} = a$), and
   \item $N_{ab}^c = N_{\ol{a}c}^b = N_{c\ol{b}}^a$ for all charges $a,b,c$.
\end{enumerate}

The last requirement, called Frobenius reciprocity, comes from the following idea. If we allow the charges to split as well as fuse, then the fusion of $a$ and $b$ to $c$ is equivalent to the splitting of $1$ in $\ol a$ and $a$ followed by the fusion of $a$ and $b$. Looking at this process we now find that $b$ splits into $\ol a$ and $c$.  The last requirement then means that $\ol a$ and $c$ can fuse to $b$ in the same number of ways that $a$ and $b$ can fuse to $c$. Note that this is a non-trivial assumption.

The mathematical structure that captures the properties of the examples above is that of a fusion ring. There are several, subtly different, definitions of fusion rings spread across the literature. To avoid confusion we will use the following definition.

\begin{definition}\label{def:fusionring}
   A $\mathbb{Z}_+$-ring is a unital ring $R$, finitely generated and free as a $\mathbb{Z}$-module, which is equiped with a distinguished basis $B$ such that $1 \in B$ and for which the structure constants $\{ N_{ab}^c \}$ are non-negative.

   A \textbf{fusion ring} $R \equiv (R,B,\ol{\cdot})$ is a $\mathbb{Z}_+$ ring with basis $B$ and a linear involution $\ol{\cdot}:a \mapsto \ol{a}$ such that $N_{a\ol{b}}^1 = \delta_{ab}$ and $N_{ab}^c = N_{\ol{a}c}^b$ for all $a,b,c \in B$.
   The size of $B$ is called the \textbf{rank of the fusion ring} and the number $\max\{N_{ab}^c\}$ is called the \textbf{multiplicity of the fusion ring}.
   Elements of a fusion ring are also called \textbf{charges} or \textbf{particles} and if $a$ is an element of a fusion ring then $\ol{a}$ is called the \textbf{dual} to $a$, or \textbf{conjugate} to $a$.
   We also say that $B$ generates $R$ and write $R = \langle B \rangle \equiv \langle \psi_1, \ldots, \psi_r \rangle$.
 \end{definition}

This definition is slightly stricter than the one given in \cite{Etingof2015} in the sense that we require Frobenius reciprocity but more general than, e.g., the one in \cite{Fuchs1994} since we allow a fusion ring to be non-commutative.
From now on we will reserve the notation $R$ for a fusion ring, $B$ for its basis, $r$ for its rank, $m$ for its multiplicity, and $N_{ab}^c$ for the structure constants.

The defining properties of a fusion ring are only a slight generalization of those that define a group. Indeed, any group ring is a fusion ring and for any fusion ring $R$, demanding that $a \times \ol{a} = 1$ for all $a \in R$ comes down to demanding that $R$ is a group ring. This means that classifying fusion rings is at least as hard as classifying finite groups and there is little hope that fusion rings will be classified soon.

Instead of attempting to classify all fusion rings, we present an algorithm to find all fusion rings of small rank and multiplicity. The algorithm, together with some of the results, is the main topic of section \ref{s:FindingFusionRings}. Section \ref{s:charactersandmoddata} deals with methods for finding fusion ring characters and modular data for commutative fusion rings. In section \ref{s:categorifiability} several categorifiability criteria from \cite{Liu2021fusion}, \cite{Liu2022triangular} and \cite{Liu2020} are stated and applied to the fusion rings found. In section \ref{s:noncomrings} we generalize the Tambara-Yamagami and Haagerup-Izumi constructions of fusion rings and discuss some of the structure of non-commutative fusion rings. The final section \ref{s:toolsforworkingwithrings} introduces the \textit{AnyonWiki}, a website containing fusion rings with their properties, and FusionRings.wl: a Wolfram Language package for working with fusion rings.

\section{Finding Fusion Rings} \label{s:FindingFusionRings}
\subsection{Algorithm}
Several methods exist for finding fusion rings using a computer.
One could, e.g., try using brute force to generate all integer rings of a certain rank and
multiplicity and filter those that do not satisfy the requirements of a fusion ring.
Even after breaking symmetry and reducing the number of variables (see subsections \ref{ss:reducingnofvars} and \ref{ss:breakingpermsym}) this method quickly becomes unfeasible, as can be seen in table \ref{tab:sizesearchspace}.
\begin{table}[h]
      \[
      \begin{array}{|c|ccccccc|}
  \hline
 & 3 & 4 & 5 & 6 & 7 & 8 & 9 \\
 \hline
 1 & 2.0\times 10^1 & 1.2\times 10^3 & 1.1\times 10^6 & 3.5\times 10^{10} & 7.3\times 10^{16} & 1.9\times 10^{25} & 1.3\times 10^{36} \\
 2 & 9.0\times 10^1 & 6.1\times 10^4 & 3.5\times 10^9 & 5.0\times 10^{16} & 5.2\times 10^{26} & 1.2\times 10^{40} & 1.8\times 10^{57} \\
 3 & 2.7\times 10^2 & 1.1\times 10^6 & 1.1\times 10^{12} & 1.2\times 10^{21} & 5.2\times 10^{33} & 3.7\times 10^{50} & 1.8\times 10^{72} \\
 4 & 6.5\times 10^2 & 9.8\times 10^6 & 9.6\times 10^{13} & 2.9\times 10^{24} & 1.4\times 10^{39} & 5.2\times 10^{58} & 7.5\times 10^{83} \\
 5 & 1.3\times 10^3 & 6.1\times 10^7 & 3.7\times 10^{15} & 1.7\times 10^{27} & 3.8\times 10^{43} & 2.3\times 10^{65} & 2.4\times 10^{93} \\
 6 & 2.5\times 10^3 & 2.8\times 10^8 & 8.0\times 10^{16} & 3.8\times 10^{29} & 2.1\times 10^{47} & 9.7\times 10^{70} & 2.6\times 10^{101} \\
 7 & 4.2\times 10^3 & 1.1\times 10^9 & 1.2\times 10^{18} & 4.1\times 10^{31} & 3.7\times 10^{50} & 7.2\times 10^{75} & 2.3\times 10^{108} \\
 8 & 6.6\times 10^3 & 3.5\times 10^9 & 1.2\times 10^{19} & 2.5\times 10^{33} & 2.7\times 10^{53} & 1.4\times 10^{80} & 3.2\times 10^{114} \\
 \hline
\end{array}

      \]
  \caption{Size of the search space of fusion rings of rank $r$ (columns) and multiplicity $m$ (rows) with 2 significant digits after reduction of the number of variables and symmetry breaking.}
  \label{tab:sizesearchspace}
\end{table}
Another strategy (see \cite{Gepner1995}) consists of simultaneously
diagonalising the fusion matrices of the particles, which is always possible if all
fusion matrices commute but this method becomes cumbersome for non-commutative fusion
rings. Instead, we built an algorithm based on a backtracking approach, or tree-search.
Before delving into the details of the algorithm we first present some general results
and techniques that were applied to make the task more tractable.

\subsection{Reducing the Number of Variables}
\label{ss:reducingnofvars}
There are relations between the structure constants that can be used to greatly reduce the number of variables. From the definition of a fusion ring, it follows that
\begin{equation}
   N_{ab}^c = \sum_{e=1}^r N_{ab}^e N_{e\ol{c}}^1 = \sum_{f=1}^r N_{af}^1 N_{b\ol{c}}^f = N_{b\ol{c}}^{\ol{a}}.
   \label{eq:pivotal1}
\end{equation}
Combined with the relations $N_{ab}^c = N_{\ol{a}c}^b$ we obtain
\begin{equation}
   N_{ab}^c = N_{\ol{a}c}^b = N_{c\ol{b}}^a = N_{b\ol{c}}^{\ol{a}} = N_{\ol{c}a}^{\ol{b}} = N_{\ol{b} \ol{a}}^{\ol{c}},
   \label{eq:pivotal}
\end{equation}
which we will call \textit{pivotal relations}.
By using the pivotal relations, a reduced set of fusion coefficients can be obtained. Since these relations depend on the number $s$ of self-dual particles, we will assume from here on that we are searching for fusion rings with a fixed value for $s$. One only needs to apply the algorithm to each of the $\left\lfloor \frac{r+1}{2} \right\rfloor$ values of $s$ to find all fusion rings of rank $r$ and multiplicity $m$.

\subsection{Breaking Permutation Symmetry}
\label{ss:breakingpermsym}
Any fusion ring is fully determined by a set of structure constants $\{N_{ab}^c\}$. Therefore
the search for fusion rings can be reduced to filling $3$D tables with natural numbers such that the defining properties like associativity, unitality, etc are apparent.
When expressing structure constants using tables, however, a labeling of the elements of
the basis $B$ is implicitly made. Any relabeling of the elements of $B$ results in a table of
structure constants that describe the same ring. In particular, for every fusion ring of rank $r$ there
are up to $r!$ different, yet equivalent, tables of structure constants. This redundancy greatly
increases the amount of work when searching for fusion rings by constructing multiplication tables.

One way to break this symmetry slightly is by numbering the elements of the basis and demanding that the
first element is the unit element. We can break the symmetry further by demanding that all self-dual
elements appear before the non-self-dual elements and all non-self-dual elements are grouped in pairs.
To break the symmetry even further we added a set of constraints on the structure constants to the set of associativity relations.
To explain the idea behind these constraints we will first assume that all particles are self-dual and later generalize to generic fusion rings. The constraints are built up by looking at a particle that is not the unit, demanding it will be the $2$nd particle and sorting the other particles \textit{based on their fusion} with this particle. Doing so gives a candidate for the $3$rd particle. Then we demand that all particles, apart from the $1$st, $2$nd, and $3$rd particle, whose positions are not uniquely fixed by their fusion with the $2$nd particle are sorted by their fusion with the $3$rd particle. A candidate for the $4$th particle is then given and we can continue this way until all particle labels are fixed. To apply this scheme we need to define an ordering of a set of particles $\{\psi_{i}\}$ that is solely based on their fusion with a given particle, say $\psi_{a}$. For this, we used the function $\iota_{a}:i \mapsto N_{ai}^i$. Practically this means that we demand that for particle $\psi_{2}$ the following set of inequalities
\begin{equation}
  \left\{ \iota_{2}(j) \leq \iota_{2}(j+1) | j = 3, 4, \ldots, r-1 \right\}
\end{equation}
holds, where for some particles $\psi_{i},\psi_k$ we still might have that $\iota_{2}(i) = \iota_{2}(k)$. Those particles are then sorted further using the values of $\iota_{3}$, which yields the following set of inequalities
\begin{equation}
  \left\{\neg (\iota_{2}(j) = \iota_{2}(j+1)) \vee (\iota_{3}(j) \leq \iota_{3}(j+1)) | j = 4, 5, \ldots, r - 1 \right\}.
\end{equation}
Applied to all particles $\{ \psi_{2}, \ldots, \psi_r \}$ we get
\begin{equation}\label{eq:sortedconstraints}
  \bigcup_{i = 2}^r \left\{ \left. \neg \left(\bigwedge_{n = 2}^{i-1} \iota_{n}(j) = \iota_{n}(j+1)\right) \vee \iota_{i}(j) \leq \iota_{i}(j+1)  \right| j =  i+1, \ldots, r - 1  \right\}.
\end{equation}
Although this scheme breaks some symmetry it is clear that there is still redundancy in the choice of the $2$nd particle. Furthermore, there might also be multiple choices for the $3$rd particle since the order of the particles, determined by fusion with the $2$nd particle, is not strict. The same is true for the $4$th particle and so on. To reduce this redundancy we demand that the $2$nd particle should be such that
\begin{equation}\label{eq:extraconstraintp2}
  \left\{ \left. \sum_{i = 1}^r \iota_{2}(i) \geq \sum_{i = 1}^r \iota_k(i) \right| k = 3, \ldots, r  \right\}.
\end{equation}
For the $3$rd, $4$th, and other particles a similar set of inequalities can be constructed but we can only compare with particles that could not be distinguished by all previous particles. It turns out that these extra constraint become so tedious that taking them into account only increases computation time. Therefore we only added constraints (\ref{eq:extraconstraintp2}) to the set of constraints.

Now assume not all particles are self-dual. In this case, we can still apply the symmetry breaking but should do so separately for the self-dual particles and non-self-dual particles. This must be done in such a way that no comparison takes place between structure constants of self-dual particles and non-self-dual particles. Note that for the non-self-dual particles, because of the pivotal relations, the constraints above will not put an ordering between dual particles. Therefore, only pairs of particles can be sorted this way.

\subsection{Backtracking}\label{ss:backtracking}
\label{ss:backtrackingalgorithm}
Under the assumption that the techniques described above have been applied one should now have (a) a reduced set of unknowns, say $\{n_i\}$, and (b) an enlarged set of constraints consisting of a set of associativity relations together with a set of inequalities (forcing a non-strict order on the particles).
The remaining procedure then consists of finding solutions to the system of constraints. It is here that a backtracking approach naturally arises since for any constraint only partial information about the values of the $n_i$ is needed to check its validity. We can therefore do a tree-search for admissible solution sets. The branches of the tree at level $j$ correspond to a choice of $n_j$. The leaves at all but the last level correspond to invalidated constraints and the leaves at the last level correspond to valid solutions.
A graphical depiction of this process is shown in figure \ref{fig:backtracking}.

\begin{figure}[h]
\centering
\includegraphics[width=\textwidth]{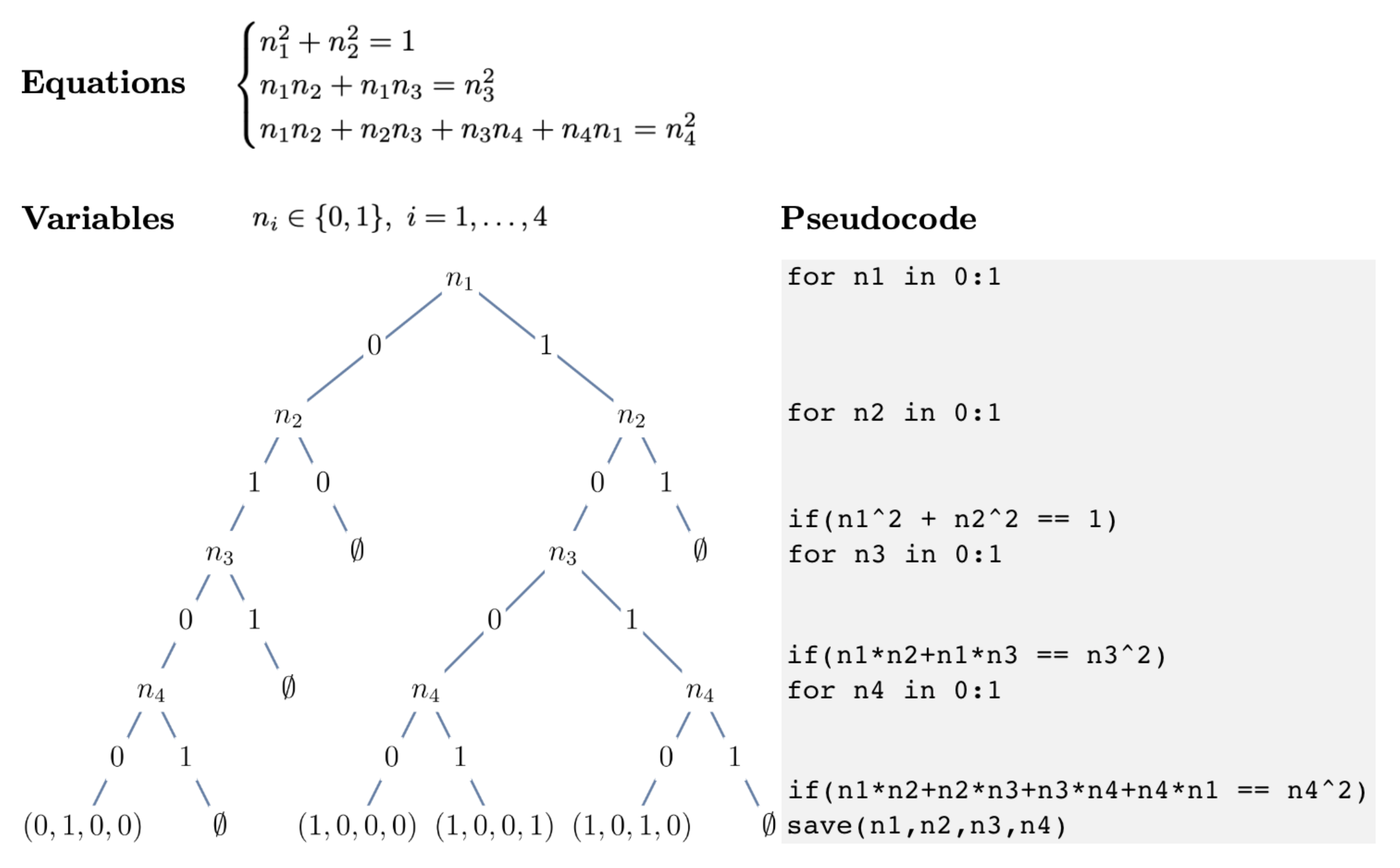}
\caption{Procedure for finding solutions to a system of constraints applied to a system of $3$ equations in $4$ unknowns with values in $\{0,1\}$. Note that the order in which the constraints are checked is vital for the performance.}
\label{fig:backtracking}
\end{figure}

Every time a constraint is violated, $m^K$ configurations (where $K$ denotes the number of remaining unknowns) are ruled out. The order in which the unknowns are given values is thus of great importance. To cut off branches of the search tree as soon as possible we sorted the unknowns in the following way.
First, we look for the constraints with the fewest number of different unknowns. If there are multiple, pick any as a first constraint. Then regard all the unknowns in this constraint as known and proceed to pick a second constraint with the least number of unknowns (thus after removing the unknowns from the first constraint from all other constraints). Keep repeating this procedure until no constraints remain.
Now group the constraints together in sets $C[i]$, $i = 1, \ldots, k$, such that every constraint in $C[i]$ requires the same set of new unknowns to be validated. For each set $C[i]$ construct a set $V[i] = \{V[i,1], V[i,2], \ldots, V[i,l_i] \}$ that contains the $l_i$ unknowns in $C[i]$, given that all the unknowns in $V[i-1], \ldots, V[1], \emptyset $ are known.
The following code then finds all fusion rings of rank $r$, multiplicity $m$, and number of self-dual particles $s$:
\begin{lstlisting}
for V[1,1] in 0:m, V[1,2] in 0:m,...,V[1,l_1] in 0:m
  if( all constr in C[1] are verified )
  for V[2,1] in 0:m, V[2,2] in 0:m,...,V[2,l_2] in 0:m
    if( all constr in C[2] are verified )
    ...
      for V[k,1] in 0:m, V[k,2] in 0:m,...,V[k,l_k] in 0:m
        if( constraints in C[k] are verified )
        saveSol( { V[1,1], V[1,2], ..., V[k,l_k] } )
\end{lstlisting}

\vspace{.3cm}
A few remarks are in place.
\begin{itemize}
  \item Because of the nature of the above algorithm, different code must be created for each rank, multiplicity, and number of self-dual particles separately. To accommodate this need we used the Wolfram Language to generate the polynomial constraints, reduce the amount of variables, break the symmetry, sort the equations and unknowns, and create, compile, and execute corresponding C source code. The Wolfram Language code we constructed and used can be found in the attached file ``SearchForFusionRings.wl''.
  \item It is important to note that it is very easy to add other specific constraints on the fusion rings. Since the constraints are sorted purely based on the number of variables, any constraint on the structure constants can easily be added and sorted with the rest. One could, e.g., add constraints on the maximum number of non-zero structure constants per particle or put a bound on several structure constants.
  \item Just like \cite{Krajecki2008} pointed out it is very hard to benchmark the code since the performance crucially depends on how efficient cache memory is used. Sometimes the CPU decides to swap between cache and RAM and computation times get multiplied with a factor (that can go up to $100$ or even higher).
\end{itemize}

\subsection{Results}
Using the algorithm described in subsection \ref{ss:backtracking}, a total of $28451$ fusion rings were found, of which $353$ are multiplicity-free and $118$ are non-commutative. These results were obtained without the use of a high-end machine. A summary of the number of rings per rank and multiplicity is given in table \ref{tab:fusionringspermultrank}. A more detailed summary for fusion rings with $r \geq 6$ is given in table \ref{tab:FRperrankmultdetailed}.
\begin{table}
   \centering
      \begin{tabular}{cc}
    &
   \begin{tabular}{c}
   Rank \\
   \end{tabular}
   \\
   \begin{tabular}{c}
      \rotatebox{90}{Multiplicity}\hspace{-1em} \\
   \end{tabular}
   &
   \begin{tabular}{|c|ccccccccc|}
      \hline
      & 1 & 2 & 3 & 4 & 5 & 6 & 7 & 8 & 9 \\
      \hline
   1  & 1 & 2 & 4  & 10  & 16       &     39  &      43  &      96  &      142 \\
   2  & 0 & 1 & 3  & 17  & 37   &     154 &      319 &      \gr{874+}& \\
   3  & 0 & 1 & 4  & 24  & 82   &     384 &     \gr{562+} &      &      \\
   4  & 0 & 1 & 6  & 45  & 134  &     872 &     \gr{1236+} &     &      \\
   5  & 0 & 1 & 5  & 55  & 209  &    \gr{533+} &      &          &      \\
   6  & 0 & 1 & 9  & 81  & 336  &    \gr{872+} &      &          &      \\
   7  & 0 & 1 & 6  & 92  & 477  &    \gr{976+} &      &          &      \\
   8  & 0 & 1 & 10 & 137 & 733  &    \gr{1672+} &     &          &      \\
   9  & 0 & 1 & 12 & 151 & 1463 &         &          &          &      \\
   10 & 0 & 1 & 9  & 186 & 1794 &         &          &      &          \\
   11 & 0 & 1 & 10 & 238 & 2283 &         &          &      &          \\
   12 & 0 & 1 & 20 & 291 & 3049 &         &          &      &          \\
   13 & 0 & 1 & 9  & 246 & \gr{1300+} &     &          &      &          \\
   14 & 0 & 1 & 13 & 340 & \gr{1323+} &     &          &      &          \\
   15 & 0 & 1 & 16 & 349 & \gr{1550+} &     &          &      &          \\
   16 & 0 & 1 & 25 & 525 & \gr{1925+} &     &          &      &          \\
   \hline
   \end{tabular}
   \\
   \end{tabular}

   \caption{Table of total number of fusion rings per rank and multiplicity. The gray numbers with a + indicate partial results from an incomplete search.}
   \label{tab:fusionringspermultrank}
\end{table}
\begin{table}
\[
   \begin{array}{|c|cccccccccccccccc|}
\hline
 & 6^0 & 6^2 & 6^4 & 7^0 & 7^2 & 7^4 & 7^6 & 8^0 & 8^2 & 8^4 & 8^6 & 9^0 & 9^2 & 9^4 & 9^6 & 9^8 \\
\hline
1 & 20 & 9 , 2   & 8 & 18 & 14 , 3 & 7 & 1 & 38 & 17 , 13 & 3 , 15 & 7 , 3 & 46 & 34 , 11 & 12 , 21 & 13 , 3 & 2 \\
2 & 13 & 37 , 2  & 102 & 2 & 32 & 86 , 5 & 194 & & & & & & & & & \\
3 & 16 & 81 , 1  & 286 & & & & & & & & & & & & & \\
4 & 17 & 151 , 1 & 703 & & & & & & & & & & & & & \\
\hline
\end{array}

\]
\caption{Table of fusion rings per multiplicity (rows) and rank (columns), where the rank is subdivided in sets $r^i$ with $i$ the number of non-self-dual particles. Numbers separated by a comma, $a,b$ indicate that there are $a$ commutative rings and $b$ non-commutative rings.}
\label{tab:FRperrankmultdetailed}
\end{table}

\begin{figure}[h!]
\centering
\includegraphics[width=\textwidth]{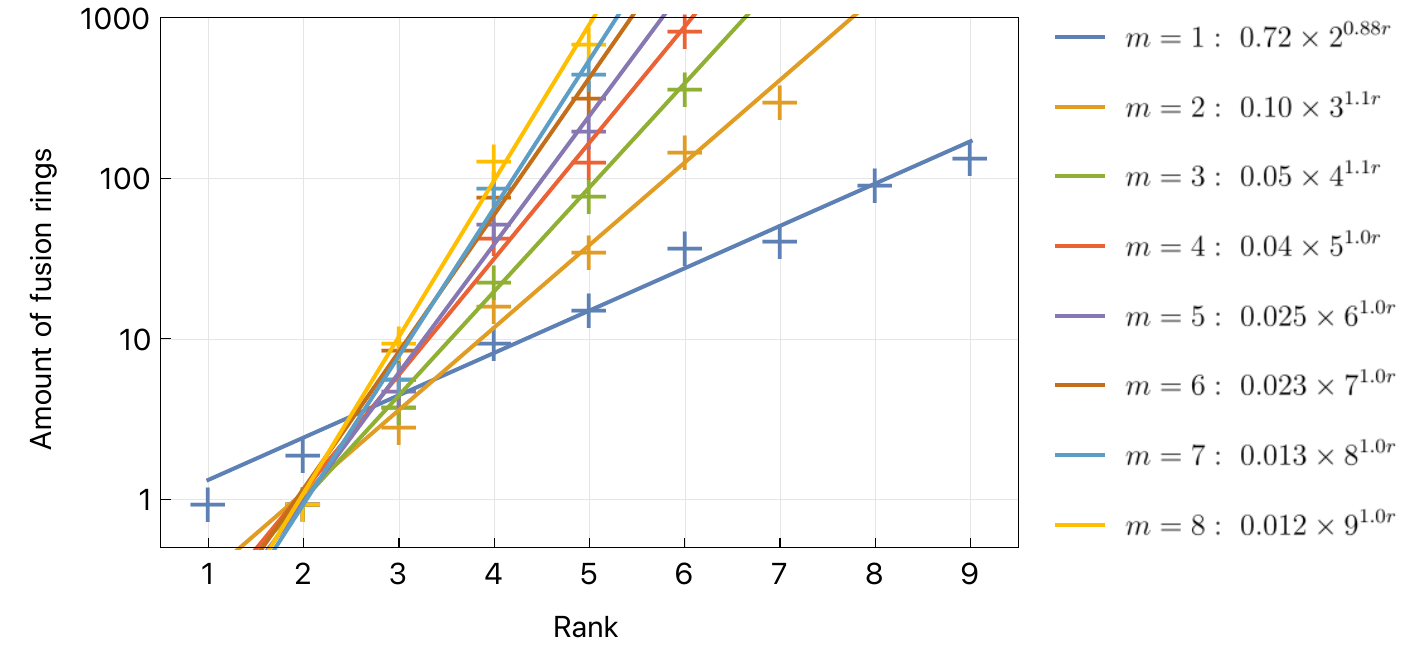}
\caption{The number of fusion rings per rank. Each line represents a least squares fit of the function $a(m+1)^{b r}$.}
\label{fig:fusionringsperrankmult}
\end{figure}
A complete list of all the fusion rings with multiplication tables and other properties can be found as a part of the FusionRings.wl Wolfram Language package, developed by the authors (see subsection \ref{ss:fusionrings.wl}). The fusion matrices can be found in the attached file ``FusionRingsMultiplicationTables''.
The number of fusion rings per rank and multiplicity, together with least-squares fits, are shown in figure \ref{fig:fusionringsperrankmult}.
\begin{remark}
   It seems that the number of fusion rings for a given multiplicity $m$ can be approximated well by a function of the form $a_{m}\times (m+1)^{b r}$, where $a_{m}$ is a constant and $b \approx 1$. Whether the number of fusion rings with multiplicity $m$ grows assymptotically as $(m+1)^{r}$ for all $m$ (or as $m$ becomes larger as well) is unknown to the authors.
\end{remark}

\section{Characters and Modular Data of Commutative Fusion Rings}\label{s:charactersandmoddata}
\subsection{Characters of commutative Fusion Rings}
Finding a character table of commutative fusion rings comes down to finding a matrix $V$ that simultaneously diagonalizes the fusion matrices $[N_k]^j_i$ \cite{Fuchs1994}. A simple way to do this is by taking a linear combination $M = \sum_{k=1}^r c_k [N_k]$, where the $c_k$ are random real numbers in an interval $[a,b]$, and finding the eigenspace of $M$. There are values of $c_k$ for which this method provides an incorrect answer but it is always easy to test whether this is the case by performing the diagonalisation. When an incorrect result is returned a new set of $c_k$ can be chosen at random until a set that works is found. The following lemma shows that the set of vectors $\vec{c}$ for which the procedure above fails is a subset of a strict sub-vector space of $\mathbb{C}^r$.
\begin{proposition} Let $\{M_k\}$ be a set of simultaneously diagonalisable $n\times n$-matrices. Let $c\in \mathbb{C}^n$ be a vector for which there exists a matrix $V$ that diagonalises $M:=\sum_{k=1}^n c_k M_k$ but does not diagonalise each $M_k$ individualy. The set of vectors $\{c\}$ for which this property holds is a subset of a strict subspace of $\mathbb{C}^n$.
\end{proposition}
\begin{proof} First we note that for any $c \in \mathbb{C}^n$ there always exists a matrix $V$ that diagonalises all $M_k$ and therefore $M$ as well. Now assume there exists a $V$ that diagonalises $M$ but not every $M_k$. We then have that
  \begin{equation}
    [V^{-1}MV]^i_j = \sum_{k=1}^n c_k [V^{-1}M_k V]^i_j = 0, \text{ for } i \neq j.
  \end{equation}
  Since not all $[V^{-1}M_k V]^i_j$ are $0$ this is a system of linear equations in the variables $c_k$ with at least $1$ non-trivial equation.
\end{proof}
Whether the character table contains
\begin{enumerate}
  \item symbolic expressions, or
  \item approximate floating point numbers
\end{enumerate}
depends on the technique for finding eigenspaces. For the first case, we were able to find exact character tables for $28227$ out of the $28333$ commutative rings. In the second case, we found character tables for all commutative fusion rings where each character is expressed with $99$ significant digits.
In both cases, we checked diagonalisation by performing matrix multiplications such that the matrix elements in the final results were correct up to $99$ significant digits.
\subsection{Modular Data}
In the context of anyon models, not all rings are considered interesting. The interesting rings are those that can be categorified to a modular tensor category (MTC) or unitary modular tensor category (UMTC).
These categories provide necessary data to make predictions about processes involving anyons such as fusion, splitting, braiding, etc. Not every fusion ring can be categorified to an MTC or UMTC, but if they can then some of the structure of their categories can be extracted without the need for categorification. Since categorification of a fusion ring can be quite hard to do, any information that can be extracted beforehand is more than welcome. In particular, any MTC or UMTC has an associated $S$-matrix and $T$-matrix that can be constructed using only the fusion ring. Therefore if no such $S$ and $T$ exist we know there can be no MTC or UMTC associated with the fusion ring. The code we used to find $S$-matrices and $T$-matrices can be found in the file ``CharactersAndModularData.wl''.

\subsubsection{Finding $S$-matrices}
An $S$ matrix associated to a fusion ring $R$ is a square, symmetric, unitary matrix that diagonalises the set of matrices $\{N_i\}_{i=1}^r$, and satisfies
\begin{IEEEeqnarray}{rCl}
  \left[ S^2\right]^i_j & = & N_{ij}^1 \\
  \left[ S \right]^1_i & = & d_i,\quad i = 1,\ldots,r,
\end{IEEEeqnarray}
where $d_i$ is the Frobenius-Perron eigenvalue of $[N_i]$.
Once the character table of a commutative fusion ring is found, it is easy to construct $S$-matrices belonging to the ring. Indeed: we know that the rows of the character table are the simultaneous eigenvectors of the fusion matrices $[N_i]$, and likewise the rows and columns of the $S$-matrix consist of simultaneous eigenvectors of the $[N_i]$. By permuting the rows of the character table and rescaling them such that $S_{1i}=S_{i1}$, one can find all valid $S$-matrices starting from the character table.
\subsubsection{Finding $T$-Matrices}
A $T$-matrix belonging to a fusion ring $R$ is a square diagonal matrix for which there exists a $\lambda\in\mathbb{C}\backslash\{0\}$ such that
\begin{equation}\label{eq:modularrepresentation}
    (ST)^3 = \lambda S^2,
\end{equation}
i.e.\ the $S$-and $T$ matrices form a representation of the modular group. Instead of solving equations (\ref{eq:modularrepresentation}) directly we can use a theorem from Vafa\cite{VAFA1988421}:
\begin{theorem}
  Let $T$ be a $T$-matrix, with diagonal entries $\theta_{i}$, belonging to a fusion ring $R$ then the following (Vafa) equations hold
  \begin{IEEEeqnarray}{rCl}\label{eq:vafa}
    \theta_1 &=& 1 \\
    \left(\theta_{i}\theta_{j}\theta_k\theta_{l}\right)^{\sum\limits_{n =1}^r N_{ij}^{\ol{n}}N_{kl}^n} &=& \prod_{n=1}^r \theta_{n}^{N_{ij}^n N_{ln}^{\ol{k}} + N_{jk}^nN_{ln}^{\ol{i}} + N_{ik}^nN_{ln}^{\ol{j}}}, \quad i,j,k,l = 1,\ldots,r,
  \end{IEEEeqnarray}
\end{theorem}
To find admissible $T$-matrices one can first solve the Vafa equations and then check whether the modularity constraint (\ref{eq:modularrepresentation}) holds. To solve the Vafa equations one can take a logarithm of both sides of the equations to obtain
\begin{equation}
  \sum\limits_{n =1}^r N_{ij}^{\ol{n}}N_{kl}^n t_{i} + N_{ij}^{\ol{n}}N_{kl}^n t_{j} + N_{ij}^{\ol{n}}N_{kl}^n t_k + N_{ij}^{\ol{n}}N_{kl}^n t_{l} - \left(N_{ij}^n N_{ln}^{\ol{k}} + N_{jk}^nN_{ln}^{\ol{i}} + N_{ik}^nN_{ln}^{\ol{j}}\right) t_{n} \in \mathbb{Z}
\end{equation}
where $t_{i} := \ln(\theta_{i})/(2\pi i)$. This is a system of linear equations with integer coefficients which can, e.g., be solved using a Smith decomposition.

From the collection of commutative fusion rings we obtained, $54$ have matching $S$-and $T$-matrices.
\section{Categorifiability}\label{s:categorifiability}
Even if a fusion ring does not have an associated MTC or UMTC one can still attempt to categorify the ring to a fusion category (FC) or unitary fusion category (UFC), which are--from a mathematical viewpoint--interesting structures in their own right. Just like the case for MTCs and UMTCs there exist useful criteria to determine whether a fusion ring admits a corresponding FC or UFC. To rule out general categorifiability for fusion rings the zero-spectrum criterion (ZSC) from \cite{Liu2020} (and also \cite{Liu2022triangular}) can be used:
\begin{criterion}\label{crit:ZSC}
  For a fusion ring $R$, if there are indices $i_{j}, 1 \leq j \leq 9,$ such that $N_{i_{4}i_{1}}^{i_{6}}$, $N_{i_{5} i_{4}}^{i_{2}}$, $N_{i_{5} i_{6}}^{i_{3}}$, $N_{i_{7} i_{9}}^{i_{1}}$ , $N_{i_{2} i_{7}}^{i_{8}}$, $N_{i_{8} i_{9}}^{i_{3}}$ are non-zero, and
  \[
  \begin{aligned}
    \sum_{k=1}^r N_{i_{4} i_{7}}^{k} N_{\ol{i_{5}} i_{8}}^{k} N_{i_{6} \ol{i_{9}}}^{k} &=0;  \text{ and}\\
    N_{i_{2} i_{1}}^{i_{3}} &=1; \text{ and}\\
    \sum_{k = 1}^r N_{i_{5} i_{4}}^{k} N_{i_{3} \ol{i_{1}}}^{k} &=1 \text { or } \sum_{k=1}^r N_{i_{2} \ol{i_{4}}}^{k} N_{i_{3} \ol{i_{6}}}^{k}=1 \text { or } \sum_{k} N_{\ol{i_{5}} i_{2}}^{k} N_{i_{6} \ol{i_{1}}}^{k}=1;  \text{ and}\\
    \sum_{k = 1}^r N_{i_{2} i_{7}}^{k} N_{i_{3} \ol{i_{9}}}^{k} &=1 \text { or } \sum_{k=1} N_{i_{8} \ol{i_{7}}}^{k} N_{i_{3} \ol{i_{1}}}^{k}=1 \text { or } \sum_{k=1}^r N_{\ol{i_{2}} i_{8}}^{k} N_{i_{1} \ol{i_{9}}}^{k}=1
  \end{aligned}
  \]
then $R$ cannot be categorified.
\end{criterion}
By applying the ZSC to the fusion rings obtained we found that $267$ of the $28451$ fusion rings can not be categorified at all. To rule out unitary categorifiablility the commutative Schur product criterion (CSPC) \cite{Liu2021fusion} can be used:
\begin{criterion}\label{crit:CSPC}
    Let $R$ be a commutative fusion ring, let $\left[N_{i}\right]$ be its fusion matrices, and let $\left[\chi_{i j}\right]$ be its character table, with $\chi_{1, i}= d_{i}$, the Frobenius-Perron eigenvalue of $[N_i]$. If $\exists\left(j_{1}, j_{2}, j_{3}\right)$ such that
  \begin{equation}
    \sum_{ i = 1 }^r \frac{ \chi_{j_{1} i } \chi_{ j_{2} i } \chi_{ j_{3} i}}{ \chi_{1 j} } < 0
  \end{equation}
  then $R$ admits no unitary categorification.
\end{criterion}
Applying this criterion to the commutative fusion rings we found it turns out that out of the $28333$ commutative fusion rings, $19471$ admit no categorification to a UFC.

\section{Some Comments on Non-commutative Fusion Rings}\label{s:noncomrings}
Out of the $28451$ fusion rings we have found, $118$ are non-commutative. Apart from $4$ exeptions (see \ref{ss:noncomfusringswithoutsubgrp}), all non-commutative fusion rings contain a non-trivial subgroup.

\subsection{Song extensions of groups}
Fusion rings that contain a subgroup are called generalized near-group fusion rings \cite{thornton2012generalized}.
The most notable of such fusion rings are the Tambara-Yamagami fusion rings \cite{TAMBARA1998692} and the Haagerup-Izumi fusion rings of groups. The structure of both of these rings can be generalized as follows.

\begin{definition}
  Let $G$ be a finite group, $T$ a finite set, and $ \sigma_l:G \times T \rightarrow T: (g,t)\mapsto\sigma_l(g,t) =:g \cdot t $ a left action of $G$ on $T$ such that $T = G \cdot t_1$ for some $t_1 \in T$ and the left stabilizer of $t_1,\ G_{t_1}^l,$ obeys $G_{t_1}^l = H \trianglelefteq  G$. Let $\tilde{g} \in G$, $n \in \mathbb{N} $ and
  \begin{itemize}
    \item $ A: G/H \rightarrow G/H $ be an automorphism such that
      \begin{IEEEeqnarray}{rCl}\label{eq:autprop}
        A^2 ([g]) &=& [\tilde{g}^{-1} g \tilde{g}], \ \forall g \in G,\text{ and} \\
        A([\tilde{g}]) &=& [\tilde{g}],
      \end{IEEEeqnarray}
      where $ [\cdot] $ denotes the canonical projection from $G$ to $G/H$,
    \item $ \Phi: T \rightarrow G/H$ be such that $ \Phi(g \cdot t_1) = [g]$, i.e. $\Phi$ maps a $t \in T$ to the class in $G/H$ that maps $t_1$ to $t$, and
    \item $ \lambda: G/H \rightarrow G $ be a lift of the elements of $G/H$ into $G$, i.e.\ $[\lambda(gH)] = g, \forall g \in G$.
  \end{itemize}
  The set $G \sqcup  T$ with the following product
    \begin{IEEEeqnarray}{rCl}
    g \times g' &=& g g' \\
    g \times t  &=& g \cdot t \\
    t \times g  &=& \lambda( \Phi(t) A(g) )  \cdot t_1  \\
    t \times t' &=& \lambda( \Phi(t) A(\Phi(t')) ) \tilde{g}^{-1} \sum_{h \in H}h + n \sum_{t \in T}t,
    \end{IEEEeqnarray}
    $\forall g,g' \in G$, $\forall t,t' \in T$, is called the $n$'th \textbf{single orbit normal group} (or \textbf{song}) extension of $G$ featuring $H$, $A$, and $\tilde{g}$ and we denoted it by $ \left[ H \trianglelefteq  G \right]^{A}_{\tilde{g}|n}$.
\end{definition}

The proof that songs are well defined and the rings they produce are fusion rings can be found in appendix \ref{ap:Proof}.
We have the following
\begin{examples}
  \begin{enumerate}
  	\item Let $ T = \{t\} $, $\tilde{g} = 1$, and $n \in \mathbb{N} $. Then $H = G$, $ A $ is trivial, and the fusion rules become
  \begin{IEEEeqnarray}{rCl}
    g_i \times g_j &=& g_i g_j \\
    g_i \times t   &=& t \times g_i = t \\
    t \times t     &=& \sum_{g \in G}g + n t.
  \end{IEEEeqnarray}
  This ring is called a near-group fusion ring, and in particular when $n = 0$ this ring is called the Tambara-Yamagami fusion ring of the group $G$: $\text{TY}(G)$.
  Such a ring is non-commutative iff the the group $G$ is non-commutative.
  Corollary \ref{cor:1particleextension} tells us that these songs capture all extensions of any group by one particle.
\item Let $G = \{g_1,\ldots,g_n\}$ be a commutative group, $T = \{t_1, g_2 \cdot t_1,\ldots,g_n\cdot t_1\} $ (and thus $ H = \{1\}$), $\tilde{g} = 1$, and $ A : g \mapsto g^{-1} $ then the fusion rules become
  \begin{IEEEeqnarray}{rCl}
    g_i \times g_j &=& g_ig_j \\
    g_i \times t_j &=& (g_ig_j)t_1 \\
    t_i \times g_j &=& (g_ig_j^{-1})t_1 \\
    t_i \times t_j &=& g_i g_j^{-1} g_0^{-1} + n \sum_{t \in T } t.
  \end{IEEEeqnarray}
  If $ n = 1 $ this ring is called the Haagerup-Izumi fusion ring of the commutative group $G$: $\text{HI}(G)$.
  Such a ring is non-commutative iff the group $G$ (seen as a group fusion ring) contains non-selfdual particles.
\item Let $G = D_3$ be the dihedral group with $6$ elements, $T = \{t_1, t_2\} $, $H = \mathbb{Z}_3$, $\tilde{g} = 1$, and $A$ trivial. The song $\left[ \mathbb{Z}_3 \trianglelefteq D_3\right]_{1|0}^{\text{Id}}$ (entry $155$ in the file ``FusionRingMultiplicationTables'') is a non-commutative categorifiable fusion ring that is not of the type $\text{TY}(G)$ or $\text{HI}(G)$ for any group $G$. This ring has categorifications since it is the Grothendieck ring of the crossed product category $\mathcal{C}_{\mathrm{TY}(\mathbb{Z}_3)}\rtimes \mathbb{Z}_2$. \cite{Etingof2015}
\item Let $\alpha:\mathbb{Z}_3 \rightarrow \mathbb{Z}_3:g\mapsto g^{-1}$. The song $\left[ \mathbb{Z}_2 \trianglelefteq \mathbb{Z}_6 \right]_{1|0}^{\alpha}$ (entry $305$ in the file ``FusionRingMultiplicationTables'') also has categorifications but is neither of the types $\text{TY}(G)$ or $\text{HI}(G)$ for any group $G$, nor is it the Grothendieck ring of a crossed product category. The files with the data of the categorifications of this and previous example are included as ancillary files ``CategorificationsFR\_8\_1\_2\_2.wdx'' and ``CategorificationsFR\_9\_1\_4\_3.wdx'' respectively.
  \end{enumerate}
\end{examples}

\subsection{Generic Non-commutative Fusion Rings}
Not all non-commutative fusion rings are of this type, however. Table \ref{tab:detailednoncompergroup} gives an overview of the multiplicity-free fusion rings per subgroup and rank.
\begin{table}
  \[
\begin{array}{|c|c|c|c|c|}
  \hline
  & 6 & 7 & 8 & 9 \\
  \hline
  \mathbb{Z}_2 &  &  & \begin{array}{ll} \text{FR}_{11}^{\text{8,1,2}}, & \text{FR}_{29}^{\text{8,1,2}}, \\ \text{FR}_{30}^{\text{8,1,2}} & \end{array} & \begin{array}{ll} \text{FR}_{32}^{\text{9,1,2}}, & \text{FR}_{38}^{\text{9,1,2}}, \\ \text{FR}_{41}^{\text{9,1,2}}, & \text{FR}_{44}^{\text{9,1,2}}, \\ \text{FR}_{20}^{\text{9,1,4}}, & \text{FR}_{31}^{\text{9,1,4}}, \\ \text{FR}_{33}^{\text{9,1,4}} & \end{array} \\
  \hline
  \mathbb{Z}_3 & \text{HI}(\mathbb{Z}_3) & \text{FR}_{15}^{\text{7,1,2}} &  & \text{FR}_{27}^{\text{9,1,4}} \\
  \hline
  \mathbb{Z}_4 &  &  &
   \begin{array}{ll}
      \text{FR}_8^{\text{8,1,2}}, & \text{HI}(\mathbb{Z}_4), \\%
      \text{FR}_8^{\text{8,1,4}}, & \text{FR}_7^{\text{8,1,6}}, \\ \relax%
      [I\trianglelefteq\mathbb{Z}_4]^{\alpha}_{\mathbf{2};1} &
   \end{array} &
   \begin{array}{ll}
      \text{FR}_{36}^{\text{9,1,2}}, & \text{FR}_9^{\text{9,1,4}}, \\ \text{FR}_{17}^{\text{9,1,4}}, & \text{FR}_{13}^{\text{9,1,6}}
   \end{array} \\
  \hline
  \mathbb{Z}_2\times \mathbb{Z}_2 &  &  & \begin{array}{ll}\text{FR}_6^{\text{8,1,2}}, & \text{FR}_9^{\text{8,1,2}}, \\ \text{FR}_{10}^{\text{8,1,2}}, & [I\trianglelefteq\mathbb{Z}_2×\mathbb{Z}_2]^{\alpha}_{\mathbf{1};1} \end{array} & \begin{array}{ll} \text{FR}_7^{\text{9,1,2}}, & \text{FR}_{22}^{\text{9,1,2}}, \\ \text{FR}_{37}^{\text{9,1,2}}, & \text{FR}_4^{\text{9,1,4}}, \\ \text{FR}_8^{\text{9,1,4}}, & \text{FR}_{22}^{\text{9,1,4}}\end{array} \\
  \hline
  \mathbb{Z}_6 &  &  &  & [ \mathbb{Z}_2 \trianglelefteq \mathbb{Z}_6]^{\alpha}_{\mathbf{1};0},  [ \mathbb{Z}_2 \trianglelefteq \mathbb{Z}_6]^{\alpha}_{\mathbf{1};1} \\
  \hline
  D_3 & D_3  &
  \begin{array}{l}
      \text{TY}(D_3), \\ \relax
      [D_3\trianglelefteq D_3]^{\mathrm{Id}}_{\mathbf{1};1}
  \end{array} & \begin{array}{ll}[\mathbb{Z}_3\trianglelefteq D_3]^{\mathrm{Id}}_{\mathbf{1};0}, & [\mathbb{Z}_3\trianglelefteq D_3]^{\mathrm{Id}}_{\mathbf{2};1}, \\ \relax [\mathbb{Z}_3\trianglelefteq D_3]^{\mathrm{Id}}_{\mathbf{1};1}, & [ \mathbb{Z}_3 \trianglelefteq D_3 ]^{\mathrm{Id}}_{\mathbf{2};0}, \\ \relax
  \text{FR}_3^{\text{8,1,2}}  & \end{array} & \text{FR}_4^{\text{9,1,2}},\text{FR}_6^{\text{9,1,4}} \\
  \hline
  D_4 &  &  & D_4  & \text{TY}(D_4),[ D_4 \trianglelefteq D_4 ]^{\mathrm{Id}}_{\mathbf{1};1} \\
  \hline
  Q &  &  & Q & \text{TY}(Q),[ Q \trianglelefteq Q ]^{\mathrm{Id}}_{\mathbf{1};1} \\
  \hline
\end{array}
\]

  \caption{Non-commutative multiplicity-free fusion rings per non-trivial largest subgroup (rows) and rank (columns). Here $I$ stands for the trivial group, HI$(G)$ for the Haagerup-Izumi fusion ring of $G$ and TY$(G)$ for the Tambara-Yamagami fusion ring of $G$. $\alpha$ denotes the, up to isomorphism, only non-trivial group automorphism, $\mathbf{1}$ the unit element, and $\mathbf{2}$ the, up to isomorphism, only non-trivial group element. For the meaning of the names $\text{FR}^{r,m,n}_i$ see appendix \ref{s:namingoffusionrings}}\label{tab:detailednoncompergroup}
\end{table}
To get some insight in the properties of generic non-commutative rings and understand some of the structure of table \ref{tab:detailednoncompergroup} it is interesting to take a look at the consequences of having a subgroup structure.


Let $G = \langle g_1 = 1, \ldots, g_n \rangle$ be a group fusion ring and consider extending it by adding elements from a finite set, say $t_1,\ldots,t_m \in T$.
Now consider a product, say $\times$, on $G \sqcup  T$ that coincides with the product on $G$, i.e.\ $g \times h = g h, \forall g,h \in G$.
Since $G$ is a group, and a fusion ring is associative and has a unique unit, the left and right action $\sigma^l_g : t \mapsto g \times t$ and $\sigma^r_g: t \mapsto t \times g$ of $G$ on $T$ must permute the elements of $T$.
Requiring that the ring $G_T := \langle g_1 = 1, \ldots, g_{n}, t_1, \ldots t_{m} \rangle$ with product $\times$ is a fusion ring puts further restrictions on $\times$:
\begin{proposition}\label{prop:structnonabrings}
   Let $G$ and $T$ be as above and write $G^l_t$ ($G^r_t$) for the left (right) stabilizer of $t$ under $G$, $[t]^l$ ($[t]^r$) for the left (right) orbit of $t$ under $G$ and $H^l_i$ ($H^r_i$) for the elements of $G/G^l_t$ ($G/G^r_t$) (with $H^l_1 = G^l_t$ and $H^r_1 = G^r_t$) then for all $t \in T, g \in G$
   \begin{enumerate}
      \item
         \begin{equation}
            G^l_t = G^r_{\ol{t}},
            \label{eq:gltisgrtdual}
         \end{equation}
      \item for all $a,b \in G \sqcup  T: N_{ab}^g \in \{0,1\}$ and
         \begin{eqnarray}
            N_{t\ol{t}}^g &=& 1 \Leftrightarrow g \in G^l_t, \text{ and}\\
            N_{\ol{t}t}^g &=& 1 \Leftrightarrow g \in G^r_t
            \label{eq:stabilizerinRHS}
         \end{eqnarray}
      \item every $\tau \in [t]^l$ can be labelled by a unique $H^l_i$ (say $\tau_{H^l_i} \equiv \tau_i$) in such a way that for all $g$ in $G$
         \begin{equation}
            N_{\tau_i \ol{t}}^g = 1 \Leftrightarrow g \in H^l_i ,
            \label{eq:cosetsinRHS}
         \end{equation}
      \item if $|T| < |G|$ then there exists a $t \in T$ with non-trivial $G^l_t$ such that $\left|[t]^l\right|$ divides $|G|$, \label{st:orbitstab1}
      \item if $|T| = |G|$ and there exists a $t \in T$ with trivial $G^l_t$, then $[t]^l = T$ and $G^l_t$ is trivial for all $t \in T$, and
      \item if $|T| > |G|$ then at least $\mathrm{mod}(|G|,|T|)$ elements of $T$ have a non-trivial left stabilizer.
   \end{enumerate}
\end{proposition}
\begin{proof} Let $e$ denote the unit of $G$.
   \begin{enumerate}
      \item For all $h$ in $G^l_t:$ $e \in \ol{t} \times t = \ol{t} \times (h \times t) = (\ol{t} \times h) \times t$ but every element of a fusion ring has a unique dual so $\ol{t} \times h = \ol{t}$ for all $h$ in $G^l_t$.
      \item Since $e \in t \times \ol{t}$, $g \in g \times t \times \ol{t}$ so $g \in G^l_t \Rightarrow g \in t \times \ol{t}$. Now assume that $g \notin G^l_t$ and $g \in t \times \ol{t}$. Then $e \in g^{-1} \times t \times \ol{t}$ so $ g^{-1} \times t = \ol{\ol{t}} = t$ which is impossible by assumption. Similar reasoning can be applied to prove the second formula.
      \item For all  $h_i$ in $H^l_i$: $\tau_i \times \ol{t} = h_i \times t \times \ol{t} \ni h_i$
   \end{enumerate}
The last three statements follow directly from the orbit stabilizer theorem.
\end{proof}

The following corollary then classifies all $1$-particle extensions of groups to fusion rings.
\begin{corollary}\label{cor:1particleextension}
   Let $G$ and $T$ be as above. If $|T| = 1$ then for each $k \in \mathbb{N} $ there exists only one extension $R$ of $G$ by $T$. Moreover $R$ is commutative if and only if $G$ is commutative and for $t \in T$ and all $g \in G$ (seen as elements of $R$)
     \begin{IEEEeqnarray}{C}
        g \times t = t \times g = t \\
        t \times t = \sum_{g\in G}g + k\ t
     \end{IEEEeqnarray}
\end{corollary}
\begin{proof}
   Since $G$ stabilizes $t$ and $t$ is self-dual it follows from (\ref{eq:stabilizerinRHS}) that $N_{tt}^g = N_{t\ol{t}}^g = 1$ for all $g$ in $G$. The only remaining degree of freedom is the value of $N_{tt}^t \in \mathbb{N} $ which can be chosen freely as is verified by checking the associativity constraints.
\end{proof}
For $|T| = 2$ commutativity of an extension also depends solely on the commutativity of the group and all extensions can be classified as well.
\begin{proposition}\label{prop:2particleextension}
  If $R$ is an extension of a non-trivial group $G = \langle g_1,\ldots,g_{n}\rangle$ by $T = \{ t_1, t_{2} \}$ then we have for all $a,b,c,d \in \mathbb{N}$ that
  \begin{enumerate}
    \item $R$ is a commutative fusion ring iff $G$ is a commutative group,
    \item if $G = G^l_{t} \ \forall t \in T$ and $t_{i} = \ol{t_{i}}$
    \begin{IEEEeqnarray}{rCl}
      t_1 \times t_1 &=& \sum_{g\in G}g + a t_1 + b t_{2} \\
      t_{2} \times t_{2} &=& \sum_{g\in G}g + c t_1 + d t_{2} \\
      t_1 \times t_{2} &=& t_{2} \times t_1 = b t_1 + c t_{2},
    \end{IEEEeqnarray}
    where: $ b(b-d) + c(c-a) = |G|$,
    \item if $G = G^l_{t} \ \forall t \in T$ and $t_{i} \neq \ol{t_{i}}$
    \begin{IEEEeqnarray}{rCl}
      t_1 \times t_1 &=& a t_1 + b t_{2} \\
      t_{2} \times t_{2} &=& b t_1 + a t_{2} \\
      t_1 \times t_{2} &=& t_{2} \times t_1 = \sum_{g \in G}g + a t_1 + a t_{2},
    \end{IEEEeqnarray}
    where: $b^2 - a^2 = |G|$,
    \item if $G \neq G^l_{t} \ \forall t \in T$ and $t_{i} = \ol{t_{i}}$
    \begin{IEEEeqnarray}{rCl}
      t_1 \times t_1 &=& t_{2} \times t_{2} = \sum_{g \in G^l_t}g + a t_1 + b t_{2} \\
      t_1 \times t_{2} &=& t_{2} \times t_1 = \sum_{g \notin  G^l_t}g + b t_1 + a t_{2}
    \end{IEEEeqnarray}
    \item if $G \neq G^l_{t} \ \forall t \in T$ and $t_{i} \neq \ol{t_{i}}$
    \begin{IEEEeqnarray}{rCl}
      t_1 \times t_1 &=& t_{2} \times t_{2} = \sum_{g \notin  G^l_t}g + a t_1 + a t_{2} \\
      t_1 \times t_{2} &=& t_{2} \times t_1 = \sum_{g \in G^l_t}g + a t_1 + a t_{2}
    \end{IEEEeqnarray}
  \item if $|G| > 2$ there are no multiplicity-free extensions $R$ of $G$ by $T$ where $G$ stabilizes every element of $T$,\label{prop:2particleextension6}
    \item if $R$ is multiplicity free then $|G|$ is even.
  \end{enumerate}
\end{proposition}
\begin{proof}
  \begin{itemize}
    \item[(1)] By setting $a = \ol{c}$ in the pivotal relation $N_{ab}^c = N_{b\ol{c}}^{\ol{a}}$ we obtain $N_{ab}^{\ol{a}} = N_{ba}^{\ol{a}}$. For both the case where $t_1 = \ol{t_1}$ and the case where $t_1 = \ol{t_2}$ this results in the relations $N_{t_1 t_2}^{t_1} = N_{t_2 t_1}^{t_1}$ and $N_{t_1 t_2}^{t_2} = N_{t_2 t_1}^{t_2}$. Therefore  $t_{i} \times t_{j} = t_{j} \times t_{i},\ i = 1,2$. We also have that $G_t^l = G_t^r, \ \forall t \in T$ (and therefore $t_{i} \times r = r \times t_{i} \forall r \in R$). Indeed: for the case where $t_{i} = \ol{t_{i}}$ this follows from equation (\ref{eq:gltisgrtdual}). If $t_{i} \neq \ol{t_{i}}$ and there would exist a $t \in T,g \in G$ with $g \in G^l_t$ but $g \notin  G^r_t$ then $t_1 \times t_1 = t_1 \times g \times t_1 = t_{2} \times t_1 \Rightarrow t_1 = \ol{t_1}$, a contradiction.
    \item[(2)$\ldots$(5)] Follows directly from applying equations (\ref{eq:stabilizerinRHS}) and (\ref{eq:cosetsinRHS}) in combination with the pivotal relations (\ref{eq:pivotal}), and the associativity constraints.
    \item[(6)] If $G$ stabilizes every element of $T$ then either $|G| = b(b-d) + c(c-a) \leq 2$ or $|G| = b^2 - a^2 \leq 1$.
    \item[(7)] If $G$ stabilizes every element of $T$ then, since $G$ is non-trivial, $|G| = 2$. If $G$ does not stabilize every element of $T$ then for any $t\in T$, $|[t]| = |T|$ and the result follows directly from statement nr. \ref{st:orbitstab1}.
  \end{itemize}
\end{proof}
A lot can be said about fusion rings with subgroups by looking at the cardinality of the stabilizers. Indeed, in proposition \ref{prop:2particleextension} we found several relations between the size of stabilizer subgroups (in that case the whole group) and the fusion coefficients. These are useful tools for finding obstructions to extend groups into fusion rings.
The following lemma and its corollaries will prove useful to exclude extensions of groups to fusion rings for more general $T$.
\begin{lemma}
Let $\tau \in T$, then for all $a \in R$
\begin{equation}
   \left| G^l_\tau \cap G^l_a \right| = \sum_{g \in G} \left( N_{\tau a}^g \right)^2 + \sum_{t \in T}\left( N_{\tau a}^t \right)^2 - \sum_{t \in T} N_{\tau t}^\tau N_{tc}^c
\label{eq:grouporderconstraint}
\end{equation}
\end{lemma}
\begin{proof}
   For any particle $\tau \in T$ and $a \in R$ we have that
   \begin{IEEEeqnarray}{C}
      \ol{\tau} \times ( \tau \times a ) = \sum_{e \in R} \sum_{d \in R }N_{\tau a}^d N_{\ol{\tau} d}^e e
   \end{IEEEeqnarray}
   while
   \begin{IEEEeqnarray}{rCl}
(\ol{\tau} \times \tau ) \times a &=& \sum_{g \in G^l_{\tau}} g \times a + \sum_{t \in T} N_{\ol{\tau} \tau}^t t \times a
\\
&=& \sum_{g \in G^l_{\tau} \cap  G^l_{a}} a + \sum_{g \in G^l_{\tau} \backslash  G^l_{a}} g \times a + \sum_{t \in T} N_{\ol{\tau} \tau}^t t \times a
\\
&=& \left| G^l_{\tau} \cap  G^l_{a} \right| a + \sum_{g \in G^l_{\tau} \backslash  G^l_{a}} g \times a + \sum_{e \in R} \sum_{t \in T} N_{\ol{\tau} \tau}^tN_{t a}^e e
   \end{IEEEeqnarray}
   By looking at the $a$'th component and using the pivotal relation $N_{ab}^c = N_{\ol{a}c}^b$ we find that
   \begin{equation}
   \left| G^l_{\tau} \cap  G^l_{a} \right| = \sum_{g \in G } \left( N_{\tau a}^g \right)^2 + \sum_{t \in T } \left( N_{\tau a}^t \right)^2 - \sum_{t \in T} N_{\tau t}^\tau N_{t a}^a
   \end{equation}
\end{proof}
This bound becomes stronger when one or multiple particles are fixed by the group $G$.
\begin{corollary}
   If $|T| \geq 2$ and $\exists\tau \in T$ for which $G^l_\tau = G$ then for any $a \in T$ with $a \neq \ol{\tau}$
   \begin{equation}
   \left| G^l_{a} \right| = \sum_{t \in T } \left( N_{\tau a}^t \right)^2 - \sum_{t \in T} N_{\tau t}^\tau N_{t a}^a \leq |T|m^2
   \end{equation}
   where $m$ is the multiplicity of the fusion ring.
   \label{cor:multiplicityconstraint}
\end{corollary}
We can now generalize proposition \ref{prop:2particleextension}(\ref{prop:2particleextension6}):
\begin{proposition}
   If  $2 \leq |T| < |G|$ and $G^l_t = G$ for all $t \in T$ then there exists no multiplicity-free extension $R$ of $G$ by $T$. In particular if $|G|$ is prime and $2 \leq |T| < |G|$ there exists no multiplicity free-extensions of $G$ by $T$.
   \label{prop:primeconstraint}
\end{proposition}
We are now in a position that several patterns in table \ref{tab:detailednoncompergroup} can easily be explained. From corollary \ref{cor:1particleextension} we conclude that the only groups with rank less than $9$ and non-commutative $1$-particle extensions are $\text{D}_3$, $\text{D}_4$ and $\text{Q}_8$. For each group there are $2$ such extensions. Corollary \ref{cor:1particleextension} also explains the absense of the $3$ commutative groups of order $8$ in the table.
Proposition \ref{prop:primeconstraint} implies that both  $\mathbb{Z}_5$ and $\mathbb{Z}_7$ have no multiplicity free non-commutative extensions with rank less than $10$ hence their absence in the table.

\subsection{Non-Commutative Fusion Rings Without Non-Trivial Subgroup}\label{ss:noncomfusringswithoutsubgrp}
Only $4$ non-commutative rings were found that contain no non-trivial subgroup. Three of these (entries $7183$, $8513$, and $10777$ in the file ``FusionRingMultiplicationTables'') are simple, i.e. contain no subring at all.
The sole non-commutative non-simple fusion ring without a proper subgroup has the following multiplication table
\[
\begin{array}{|llllll|}
\hline
\mathbf{1} & \mathbf{2} & \mathbf{3} & \mathbf{4} & \mathbf{5} & \mathbf{6} \\
\mathbf{2} & \mathbf{1}+\mathbf{2} & \mathbf{6} & \mathbf{4}+\mathbf{5} & \mathbf{4} & \mathbf{3}+\mathbf{6} \\
\mathbf{3} & \mathbf{5} & \mathbf{1}+\mathbf{3} & \mathbf{4}+\mathbf{6} & \mathbf{2}+\mathbf{5} & \mathbf{4} \\
\mathbf{4} & \mathbf{4}+\mathbf{6} & \mathbf{4}+\mathbf{5} & \mathbf{1}+\mathbf{2}+\mathbf{3}+2 \ \mathbf{4}+\mathbf{5}+\mathbf{6} & \mathbf{3}+\mathbf{4}+\mathbf{5}+\mathbf{6} & \mathbf{2}+\mathbf{4}+\mathbf{5}+\mathbf{6} \\
\mathbf{5} & \mathbf{3}+\mathbf{5} & \mathbf{4} & \mathbf{2}+\mathbf{4}+\mathbf{5}+\mathbf{6} & \mathbf{4}+\mathbf{6} & \mathbf{1}+\mathbf{3}+\mathbf{4} \\
\mathbf{6} & \mathbf{4} & \mathbf{2}+\mathbf{6} & \mathbf{3}+\mathbf{4}+\mathbf{5}+\mathbf{6} & \mathbf{1}+\mathbf{2}+\mathbf{4} & \mathbf{4}+\mathbf{5} \\
\hline
\end{array}
\]
The elements $\mathbf{2}$ and $\mathbf{3}$ generate Fibonacci subrings. Elements $\mathbf{5} = \mathbf{3} \times \mathbf{2}$ and $\mathbf{6} = \mathbf{2} \times \mathbf{3}$ can be regarded as couples of Fibonacci particles and $\mathbf{4} = \mathbf{2} \times \mathbf{3} \times \mathbf{2} = \mathbf{3} \times \mathbf{2} \times \mathbf{3}$ as a tripple of Fibonacci particles. Indeed, the above fusion ring is completely determined by the generators $\mathbf{1}, \mathbf{2}$ and $\mathbf{3}$ together with the relations $\mathbf{2}^2 = \mathbf{1} + \mathbf{2}, \mathbf{3}^2 = \mathbf{1} + \mathbf{3}$, and $\mathbf{2} \times \mathbf{3} \times \mathbf{2} = \mathbf{3} \times \mathbf{2} \times \mathbf{3}$. The structure of the ring thus corresponds to that of a Hecke algebra with generators $\mathbf{2}$ and $\mathbf{3}$.
\section{Tools for Working With Fusion Rings}\label{s:toolsforworkingwithrings}
\subsection{Anyonwiki}\label{ss:anyonwiki}
In order to provide an overview of all known fusion rings and their properties, a MediaWiki website called Anynonwiki was launched. The anyonwiki contains pages on many multiplicity-free fusion rings we found, together with some standard statistics such as multiplication tables, the quantum dimensions of both fusion rings and individual particles, and info on categorification. The goal of the Anyonwiki is to be a central repository of data on fusion rings and anyon models that is accessible to non-experts as well. The idea is to embrace input from many researchers in the field. We invite researchers in the field to gain write access to the various pages (at the moment, we are cautious not to allow anyone on the internet to edit pages to protect the quality and maintain security). Currently, the AnyonWiki is hosted at \url{http://www.thphys.nuim.ie/AnyonWiki/index.php/Main_Page}.

\subsection{FusionRings.wl}\label{ss:fusionrings.wl}
To get an overview of the properties of fusion rings a Wolfram Language package was constructed. Much care is taken to provide a user-friendly interface that allows one to easily experiment with fusion rings and their properties. Some of the capabilities of the package are listed below.
\begin{itemize}
   \item Create fusion rings via multiplication tables, groups, products of known fusion rings, or simply work with a fusion ring from the database that is included. Functions are included to create fusion rings of special type: $\mathbb{Z}_n$, $\text{PSU}(2)_k$, $\text{SU}(2)_k$, $\text{HI}(G)$, and $\text{TY}(G)$.
   \item Fusion rings can be checked for equivalence, decomposed as a product of other fusion rings, and have their elements named, sorted, and permuted.
   \item Compute various properties of both rings and elements such as sub-fusion rings, quantum dimensions, fusion ring automorphisms, modular data, character tables, etc.
   \item Perform arithmetic with the elements of any fusion ring, or multiple fusion rings at the same time.
\end{itemize}
Every function is documented, and the source code, together with a notebook containing various examples, can be found at \url{https://github.com/gert-vercleyen/FusionRings}.
\section{Acknowledgements}
The authors would like to thank Dr. Sebastien Palcoux for helpful discussions.
J.K.S. acknowledges financial support from Science Foundation Ireland through Principal Investigator Awards 12/IA/1697 and 16/IA/4524. G. Vercleyen acknowledges financial support from Science Foundation Ireland Principal Investigator Award SFI/16/IA/4524.

\bibliographystyle{plain}
\bibliography{Bibliography.bib}

\appendix
\section{Proof Regarding Song Extensions}\label{ap:Proof}
\begin{proposition}
   Let $G$ be a finite group, $T$ a finite set, and $ \sigma_l:G \times T \rightarrow T: (g,t)\mapsto\sigma_l(g,t) =:g \cdot t $ a left action of $G$ on $T$ such that $T = G \cdot t_1$ for some $t_1 \in T$ and the left stabilizer of $t_1,\ G_{t_1}^l,$ obeys $G_{t_1}^l = H \trianglelefteq  G$. Let $\tilde{g} \in G$, $n \in \mathbb{N} $ and
   \begin{itemize}
     \item $ A: G/H \rightarrow G/H $ be an automorphism such that
      \begin{IEEEeqnarray}{rCl}
         A^2 ([g]) &=& [\tilde{g}^{-1} g \tilde{g}], \ \forall g \in G,\text{ and} \\
         A([\tilde{g}]) &=& [\tilde{g}],
      \end{IEEEeqnarray}
      where $ [\cdot] $ denotes the canonical projection from $G$ to $G/H$,
     \item $ \Phi: T \rightarrow G/H$ be such that $ \Phi(g \cdot t_1) = [g]$, i.e. $\Phi$ maps a $t \in T$ to the class in $G/H$ that maps $t_1$ to $t$, and
     \item $ \lambda: G/H \rightarrow G $ be a lift of the elements of $G/H$ into $G$, i.e.\ $[\lambda(gH)] = g, \forall g \in G$.
   \end{itemize}
   Define $ \sigma_r:T \times G \rightarrow T: \sigma_r(t,g) = \lambda(\Phi(t) A([g])) \cdot t_1 =: t \cdot g$, then
   \begin{enumerate}
     \item $ \sigma_r$ is a transitive right action of $G$ on $T$ with $G_t^r = H \ \forall t \in T$ which is compatible with the left action $ \sigma_l$ in the sense that $ g \cdot (t \cdot g') = (g \cdot t) \cdot g'$, and
     \item the set $G \sqcup  T$ with the following product
       \begin{IEEEeqnarray}{rCl}
       g \times g' &=& g g' \\
       g \times t  &=& g \cdot t \\
       t \times g  &=& t \cdot g' \\
       t \times t' &=& \lambda( \Phi(t) A(\Phi(t')) \tilde{g}^{-1} \sum_{h \in H}h + n \sum_{t \in T}t,
       \end{IEEEeqnarray}
       where $n \in \mathbb{N} $ is a fixed natural number and $g,g' \in G$, $t,t' \in T$, is a fusion ring. Moreover this fusion ring is independent of the specific choice of $ \lambda$.
   \end{enumerate}
\end{proposition}
\begin{proof}
  We first note that if $f_1,f_2:G^{i} \times T^{j} \rightarrow G$ then for all $g \in G$, $ [f_1(g_1,\ldots,g_i,t_1,\ldots,t_j)] = [f_2(g_1,\ldots,g_i,t_1,\ldots,t_j)]$ implies that
  \begin{IEEEeqnarray}{rCl}
    f_1(g_1,\ldots,g_i,t_1,\ldots,t_j) g \cdot t &=& f_2(g_1,\ldots,g_i,t_1,\ldots,t_j)g \cdot t,\\
    t \cdot g f_1(g_1,\ldots,g_i,t_1,\ldots,t_j)  &=& t \cdot g f_2(g_1,\ldots,g_i,t_1,\ldots,t_j),
  \end{IEEEeqnarray}
  and
  \begin{equation}\label{}
    f_1(g_1,\ldots,g_i,t_1,\ldots,t_j)g \sum_{h \in H}h = f_2(g_1,\ldots,g_i,t_1,\ldots,t_j)g \sum_{h \in H}h.
  \end{equation}
  In particular this means that both $\lambda$ and $\lambda \circ  A$ are homomorphisms as long as their values act on $T$ or $ \sum_{h \in H}h$.

  To reduce the number of parentheses we will write $ \alpha := \lambda \circ  A$, $ \Phi := \lambda \circ  \Phi$, $ \alpha[g] := \alpha([g])$ and $ \lambda[g] := \lambda([g])$ for all $g \in G$. The proof consists of two main parts.

  \noindent
  \textbf{Part 1: $\sigma_r$ defines a compatible right action.} First we note that $G_{t_1}^l = H$ implies that $G_t^l = H$ since for any $t \in T$ there exists a $g \in G$ such that $t = g \cdot t_1$ so $h \cdot t = hg \cdot t_1 = g g^{-1} h g \cdot t_1 = g \cdot t_1$ since $H \trianglelefteq  G$. Similarly $G^r_{t_1}= H \Longrightarrow  G^r_{t} = H, \forall t \in T$.
Let $1$ be the unit of $G$ then $t \cdot 1 = (\Phi(t) \alpha[1]) \cdot t_1 = (\Phi(t) \lambda[1]) \cdot t_1 = t$ and
  \begin{IEEEeqnarray}{rCl}
    (t \cdot g) \cdot g'
    &=& ( \Phi(t) \alpha[g]  \cdot t_1 ) \cdot g' \\
    &=& \Phi( (\Phi(t) \alpha[g] ) \cdot t_1 ) \alpha[g'] \cdot t_1 \\
    &=& \lambda([\Phi(t) \alpha[g]] \Phi(t_1)) \alpha[g'] \cdot t_1 \\
    &=& \lambda[\Phi(t)] \lambda[\alpha[g]] \lambda[1] \alpha[g'] \cdot t_1 \\
    &=& \Phi(t) \alpha[g] \alpha[g'] \cdot t_1 \\
    &=& t \cdot (gg')
  \end{IEEEeqnarray}

so $ \sigma_r $ is a right action. Compatibility with the left action then follows from
\begin{IEEEeqnarray}{rCl}
  (g \cdot t) \cdot g'
  &=& \Phi(g \cdot t) \alpha[g'] \cdot t_1 \\
  &=& \lambda([g]\Phi(t)) \alpha[g'] \cdot t_1 \\
  &=& \lambda[g] \lambda(\Phi(t)) \alpha[g'] \cdot t_1\\
  &=& g \Phi(t) \alpha[g'] \cdot t_1 \\
  &=& g \cdot (\Phi(t) \alpha[g'] \cdot t_1) \\
  &=& g \cdot ( t_1 \cdot g').
\end{IEEEeqnarray}

\noindent
\textbf{Part 2: $R$ is a fusion ring.} To prove that $R$ is a fusion ring we first prove that the product is associative, i.e., that the following eight equations are satisfied:

\begin{IEEEeqnarray}{rCl}
  ( g \times g' ) \times g'' &=& g \times ( g' \times g'' ), \label{eq:ass1} \\
  ( g \times g' ) \times t   &=& g \times ( g' \times t ), \label{eq:ass2} \\
  ( g \times t ) \times g'   &=& g \times ( t \times g' ), \label{eq:ass3} \\
  ( t \times g ) \times g'   &=& t \times ( g \times g' ), \label{eq:ass4} \\
  ( g \times t ) \times t'   &=& g \times ( t \times t' ), \label{eq:ass5} \\
  ( t \times g ) \times t'   &=& t \times ( g \times t' ), \label{eq:ass6} \\
  ( t \times t' ) \times g   &=& t \times ( t' \times g ), \label{eq:ass7} \\
  ( t \times t' ) \times t'' &=& t \times ( t' \times t'' ). \label{eq:ass8} \\
\end{IEEEeqnarray}
Equations (\ref{eq:ass1}),  (\ref{eq:ass2}), (\ref{eq:ass3}), and (\ref{eq:ass4}) are satisfied because the left-and right multiplication by group elements are compatible actions. To prove the other four equations it will be useful to define
\begin{IEEEeqnarray}{rCl}
  U &:=& \tilde{g}^{-1} \sum_{h \in H}h,\quad \text{ and }\\
  V &:=& n \sum_{t \in T}t.
\end{IEEEeqnarray}
Equation (\ref{eq:ass5}) is satisfied since
\begin{IEEEeqnarray}{rCl}
  (g \times t) \times t'
  &=& \Phi(g \cdot t) \alpha(\Phi(t')) U + V \\
  &=& g \Phi(t) \alpha(\Phi(t')) U + V \\
  &=& g \times ( \Phi(t) \alpha( \Phi(t'))U + V ) \\
  &=& g \times (t \times t').
\end{IEEEeqnarray}
Equation (\ref{eq:ass6}) is satisfied since
\begin{IEEEeqnarray}{rCl}
  (t \times g) \times t'
  &=& (\Phi(t) \alpha[g] \cdot t_1) \times t' \\
  &=& \Phi( \Phi(t) \alpha[g] \cdot t_1) \alpha(\Phi(t')) U + V \\
  &=& \Phi(t)\alpha[g] \Phi(t_1) \alpha(\Phi(t')) U + V \\
  &=& \Phi(t)\alpha[g] \alpha(\Phi(t')) U + V \\
  &=& \Phi(t)\alpha( [g]  \Phi(t')) U + V \\
  &=& \Phi(t)\alpha[ \lambda( [g] \Phi(t'))] U + V \\
  &=& \Phi(t)\alpha[\Phi(g \cdot t' )] U + V\\
  &=& \Phi(t)\alpha( \Phi( g \cdot t')) U + V \\
  &=& t \times (g \cdot t') \\
  &=& t \times (g \times t').
\end{IEEEeqnarray}
Equation (\ref{eq:ass7}) is satisfied since
\begin{IEEEeqnarray}{rCl}
  (t \times t') \times g
  &=& (\Phi(t) \alpha(\Phi(t'))U+V) \cdot g \\
  &=& \Phi(t) \alpha(\Phi(t')) U g + V \\
  &=& \Phi(t) \alpha(\Phi(t')) \tilde{g}^{-1} g \sum_{h \in H}H + V
\end{IEEEeqnarray}
while
\begin{IEEEeqnarray}{rCl}
  t \times (t' \times g)
  &=& t \times (\Phi(t) \alpha[g] \cdot t_1 ) \\
  &=& \Phi(t) \alpha(\Phi(\Phi(t) \alpha[g] \cdot t_1)) U + V \\
  &=& \Phi(t) \alpha[ \Phi(\Phi(t) \alpha[g] \cdot t_1)] U + V \\
  &=& \Phi(t) \alpha[ \Phi(t)  \alpha[g] ] U + V \\
  &=& \Phi(t) \alpha[\Phi(t)] \alpha[\alpha[g]] U+V \\
  &=& \Phi(t) \alpha(\Phi(t)) \lambda((A)^2[g]) U + V \\
  &=& \Phi(t) \alpha(\Phi(t)) \lambda[\tilde{g}^{-1} g \tilde{g}]U+V\\
  &=& \Phi(t) \alpha(\Phi(t)) \tilde{g}^{-1} g \sum_{h \in H}h + V
\end{IEEEeqnarray}
Equation (\ref{eq:ass8}) is satisfied since
\begin{IEEEeqnarray}{l}
  t \times (t' \times t'')\\
= t \times \left(\Phi(t') \alpha(\Phi(t''))\tilde{g}^{-1}\sum_{h \in H}h +n \sum_{\tau \in T}\tau\right) \\
= \Phi(t) \sum_{h \in H} \alpha[\Phi(t') \alpha(\Phi(t'')) \tilde{g}^{-1} h )] \cdot t_1 + n \sum_{\tau \in T}\left(\Phi(t)  \alpha(\Phi(\tau)) U + V\right)\\
= \Phi(t) \sum_{h \in H}\alpha([\Phi(t')][\alpha(\Phi(t''))] [\tilde{g}^{-1}] [h]) \cdot t_1 + n \Phi(t) \sum_{\tau \in T}\sum_{h \in H}\alpha(\Phi(\tau)) \tilde{g}^{-1} h + n \sum_{\tau \in T}V\\
= |H| \Phi(t)\alpha(\Phi(t'))) \alpha([\alpha(\Phi(t''))]) \alpha[\tilde{g}^{-1} ] \cdot t_1 + n \Phi(t)\sum_{g \in G}g + n|T|V \\
= |H| \Phi(t) \alpha(\Phi(t')) \lambda([\tilde{g}^{-1} \Phi(t'')\tilde{g}]) \alpha[\tilde{g}^{-1}] \cdot t_1 + n \Phi(t)\sum_{g \in G}g + n|T|V \\
= |H| \Phi(t) \alpha(\Phi(t')) \tilde{g}^{-1} \Phi(t'') \cdot t_1 +n \Phi(t)\sum_{g \in G}g + n|T|V,
\end{IEEEeqnarray}
while
\begin{IEEEeqnarray}{l}
  (t \times t') \times t'' \\
= (\Phi(t) \alpha(\Phi(t'))U + V) \times t'' \\
= \Phi(t) \alpha(\Phi(t')) \tilde{g}^{-1} \sum_{h \in H}h t'' + n \sum_{\tau \in T}\tau \times t'' \\
= |H|\Phi(t)\alpha(\Phi(t'))\ \tilde{g}^{-1}\ \Phi(t'') \cdot t_1 + n \sum_{\tau \in T}\Phi(\tau)\alpha(\Phi(t''))\tilde{g}^{-1} \sum_{h \in H}h + n \sum_{\tau \in T} V \\
= |H|\Phi(t)\alpha(\Phi(t'))\tilde{g}^{-1} \Phi(t'') \cdot t_1 + n \Phi(t)\sum_{g \in G}g + n|T|V.
\end{IEEEeqnarray}
The identity $1$ of $G$ is still an identity of $R$ and it is clearly still unique. Every element of $G$ retains its unique dual element. For the elements of $T$ the dual elements are uniquely fixed by $\tilde{g}$. Indeed, since
\begin{IEEEeqnarray}{rCl}
   t_1 \times t_1
   &=& \tilde{g}^{-1} \sum_{h \in H}h + V \\
   &=& \tilde{g}^{-1}\left(\sum_{h \in H}h + V\right) \\
   &=& \left(\sum_{h \in H}h + V\right)\tilde{g}^{-1}
\end{IEEEeqnarray}
$\tilde{g} \cdot t_1 = t_1 \cdot \tilde{g}$ is a two sided dual of $t_1$. Since for any $t \in T, t = \Phi(t)\cdot t_1$, we have that
\begin{IEEEeqnarray}{rCl}\label{eq:dualt}
   \ol{t} &=& \tilde{g} \cdot t_1 \cdot (\Phi(t))^{-1}\ \forall t \in T.
\end{IEEEeqnarray}   The dual of any $t \in T$ is unique since every element of $T$ can be obtained from $t_1$ by group multiplication and from (\ref{eq:dualt}) it follows directly that the dual of the dual is the original element itself.
\end{proof}

\section{Naming of Fusion Rings}\label{s:namingoffusionrings}
Many of the fusion rings are yet unnamed, which is inconvenient for referencing specific rings. To resolve this issue we constructed a naming scheme that uniquely characterizes a fusion ring by using $3$ to $4$ natural numbers. Every fusion ring is denoted by $\text{FR}^{r,m,n}_i$ where
\begin{itemize}
   \item $r$ denotes the rank of the fusion ring,
   \item $m$ equals the multiplicity of the fusion ring, i.e. the largest structure constant, and
   \item $n$ denotes the number of particles that are not self-dual.
\end{itemize}
Every triple $(r,m,n)$ determines a finite list of fusion rings which can be sorted in a canonical way. The number $i$ then denotes the position of the fusion ring in such a list. The canonical order which we implemented is based on a unique number assigned to each fusion ring that can be computed as follows. First, we sort the fusion rings by multiplicity, rank, number of non-self-dual particles, and number of non-zero structure constants. If there is any remaining ambiguity we sort further by calculating a unique number for each fusion ring as follows.
\begin{enumerate}
   \item First, the elements of the rings themselves are sorted. All self-dual particles are grouped and the non-self-dual particles are grouped in dual pairs. The self-dual particles appear before the pairs and are sorted by increasing quantum dimension. The pairs are sorted on increasing maximum quantum dimension.
   \item Let $N_{a,b}^c$ denote the structure constants of a multiplication table of a fusion ring. For each permutation of the elements of the ring (where the identity is kept fixed, and the previous ordering is kept fixed) a list of digits can be created by ordering the structure constants $N_{a,b}^c$ lexicographically on $a,b,c$:
\[
   N_{1,1}^1,N_{1,1}^2,\ldots,N_{1,2}^1,N_{1,2}^2,\ldots,N_{r,r}^r.
\]
\item By regarding each of these lists as a number with digits given by the list elements, a unique number is assigned to each permutation of the same ring. Moreover, these numbers are unique for each different fusion ring.
\item Taking the maximum value of the numbers generated per ring, a unique number is then assigned to each fusion ring.
\end{enumerate}
If $ m=1$ we omit this value and just write $\text{FR}^{r,n}_i$.

\section{List of multiplicity-free fusion rings up to rank $9$}
A list of multiplicity-free fusion rings of rank up to $9$ is given in the table below. Here $n$ denotes the number of non-self-dual particles, $N$ the number of non-zero structure constants, $D_{FP}^2$ the sum of the squares of the quantum dimensions of the elements of the ring, ZSC whether the fusion ring admits no categorification due the zero spectrum criterion, and CSPC whether the ring admits no unitary categorification due to the commutative Schur product criterion. Note that Abelian means that the ring is commutative, not that the Anyons are Abelian anyons (which means the ring is a group ring).
\begin{longtable}{lllllllllll}
   \hline
   Name & Common Name & $r$ & $n$ & $N$ & $D_{\text{FP}}$ & $D_{\text{FP}}^2$ & Abelian & Modular & ZSC & CSPC \\
 \hline
 $\text{FR}_1^{\text{1,0}} $ & $\text{Trivial} $ & 1 & 0 & 1 & 1. & 1. & True & True & False & False \\
 $ \text{FR}_1^{\text{2,0}} $ & $\mathbb{Z}_2$ & 2 & 0 & 4 & 1.41421 & 2. & True & True & False & False \\
 $ \text{FR}_2^{\text{2,0}} $ & $\text{Fib} $ & 2 & 0 & 5 & 1.90211 & 3.61803 & True & True & False & False \\
 $ \text{FR}_1^{\text{3,0}} $ & $\text{Ising} $ & 3 & 0 & 10 & 2. & 4. & True & True & False & False \\
 $ \text{FR}_2^{\text{3,0}} $ & $\left.\text{Rep(}D_3\right)$ & 3 & 0 & 11 & 2.44949 & 6. & True & False & False & False \\
 $ \text{FR}_3^{\text{3,0}} $ & $\text{PSU}(2)_5$ & 3 & 0 & 14 & 3.04892 & 9.2959 & True & True & False & False \\
 $ \text{FR}_1^{\text{3,2}} $ & $ \mathbb{Z}_3$ & 3 & 2 & 9 & 1.73205 & 3. & True & True & False & False \\
 $ \text{FR}_1^{\text{4,0}} $ & $ \mathbb{Z}_2\times \mathbb{Z}_2$ & 4 & 0 & 16 & 2. & 4. & True & True & False & False \\
 $ \text{FR}_2^{\text{4,0}} $ & $ \text{SU(2})_3$ & 4 & 0 & 20 & 2.68999 & 7.23607 & True & True & False & False \\
 $ \text{FR}_3^{\text{4,0}} $ & $ \left.\text{Rep(}D_5\right)$ & 4 & 0 & 22 & 3.16228 & 10. & True & False & False & False \\
 $ \text{FR}_4^{\text{4,0}} $ & $ \text{PSU(2})_6$ & 4 & 0 & 24 & 3.69552 & 13.6569 & True & False & False & False \\
 $ \text{FR}_5^{\text{4,0}} $ & $ \text{Fib$\times $Fib} $ & 4 & 0 & 25 & 3.61803 & 13.0902 & True & True & False & False \\
 $ \text{FR}_6^{\text{4,0}} $ & $ \text{PSU(2})_7$ & 4 & 0 & 30 & 4.38571 & 19.2344 & True & True & False & False \\
 $ \text{FR}_1^{\text{4,2}} $ & $ \mathbb{Z}_4$ & 4 & 2 & 16 & 2. & 4. & True & True & False & False \\
 $ \text{FR}_2^{\text{4,2}} $ & $ \text{Potts} $ & 4 & 2 & 18 & 2.44949 & 6. & True & False & False & False \\
 $ \text{FR}_3^{\text{4,2}} $ & $ \left.\text{Fib(}\mathbb{Z}_3\right)$ & 4 & 2 & 19 & 2.88145 & 8.30278 & True & False & False & False \\
 $ \text{FR}_4^{\text{4,2}} $ & $ \text{Pseudo PSU(2})_6$ & 4 & 2 & 24 & 3.69552 & 13.6569 & True & False & False & False \\
 $ \text{FR}_1^{\text{5,0}} $ & $ \left.\text{Rep(}D_4\right)$ & 5 & 0 & 28 & 2.82843 & 8. & True & False & False & False \\
 $ \text{FR}_2^{\text{5,0}} $ & $ \left.\text{Fib(}\mathbb{Z}_2\times \mathbb{Z}_2\right) $ & 5 & 0 & 29 & 3.24985 & 10.5616 & True & False & False & False \\
 $ \text{FR}_3^{\text{5,0}} $ & $ \text{SU(2})_4$ & 5 & 0 & 35 & 3.4641 & 12. & True & True & False & False \\
 $ \text{FR}_4^{\text{5,0}} $ & $ \left.\text{Rep(}D_7\right) $ & 5 & 0 & 37 & 3.74166 & 14. & True & False & False & False \\
 $ \text{FR}_5^{\text{5,0}} $ & $ \text{} $ & 5 & 0 & 39 & 4.07499 & 16.6056 & True & False & True & True \\
 $ \text{FR}_6^{\text{5,0}} $ & $ \left.\text{Rep(}S_4\right) $ & 5 & 0 & 43 & 4.89898 & 24. & True & False & False & False \\
 $ \text{FR}_7^{\text{5,0}} $ & $ \text{PSU(2})_8$ & 5 & 0 & 45 & 5.11667 & 26.1803 & True & False & False & False \\
 $ \text{FR}_8^{\text{5,0}} $ & $ \text{} $ & 5 & 0 & 48 & 5.57605 & 31.0923 & True & False & False & False \\
 $ \text{FR}_9^{\text{5,0}} $ & $ \text{} $ & 5 & 0 & 53 & 5.49019 & 30.1421 & True & False & True & True \\
 $ \text{FR}_{10}^{\text{5,0}} $ & $ \text{PSU(2})_9$ & 5 & 0 & 55 & 5.88612 & 34.6464 & True & True & False & False \\
 $ \text{FR}_1^{\text{5,2}} $ & $ \left.\text{TY(}\mathbb{Z}_4\right) $ & 5 & 2 & 28 & 2.82843 & 8. & True & False & False & False \\
 $ \text{FR}_2^{\text{5,2}} $ & $ \left.\text{Fib(}\mathbb{Z}_4\right) $ & 5 & 2 & 29 & 3.24985 & 10.5616 & True & False & False & False \\
 $ \text{FR}_3^{\text{5,2}} $ & $ \text{Pseudo SU(2})_4$ & 5 & 2 & 35 & 3.4641 & 12. & True & False & False & False \\
 $ \text{FR}_4^{\text{5,2}} $ & $ \left.\text{Pseudo Rep(}S_4\right) $ & 5 & 2 & 43 & 4.89898 & 24. & True & False & False & False \\
 $ \text{FR}_5^{\text{5,2}} $ & $ \text{} $ & 5 & 2 & 48 & 5.57605 & 31.0923 & True & False & False & False \\
 $ \text{FR}_1^{\text{5,4}} $ & $ \mathbb{Z}_5$ & 5 & 4 & 25 & 2.23607 & 5. & True & True & False & False \\
 $ \text{FR}_1^{\text{6,0}} $ & $ \mathbb{Z}_2\text{$\times $Ising} $ & 6 & 0 & 40 & 2.82843 & 8. & True & True & False & False \\
 $ \text{FR}_2^{\text{6,0}} $ & $ \left.\mathbb{Z}_2\text{$\times $Rep(}D_3\right) $ & 6 & 0 & 44 & 3.4641 & 12. & True & False & False & False \\
 $ \text{FR}_3^{\text{6,0}} $ & $ [\mathbb{Z}_2 \trianglelefteq \mathbb{Z}_2\times \
\mathbb{Z}_2]_{\mathbf{1}|1}^{\text{Id}} $ & 6 & 0 & 48 & 4.35066 & 18.9282 & True & False & False & False \\
 $ \text{FR}_4^{\text{6,0}} $ & $ \text{TriCritIsing} $ & 6 & 0 & 50 & 3.80423 & 14.4721 & True & True & False & False \\
 $ \text{FR}_5^{\text{6,0}} $ & $ \left.\text{Fib$\times $Rep(}D_3\right) $ & 6 & 0 & 55 & 4.65921 & 21.7082 & True & False & False & False \\
 $ \text{FR}_6^{\text{6,0}} $ & $ \text{SU(2})_5$ & 6 & 0 & 56 & 4.31182 & 18.5918 & True & True & False & False \\
 $ \text{FR}_7^{\text{6,0}} $ & $ \text{} $ & 6 & 0 & 56 & 4.24264 & 18. & True & False & False & False \\
 $ \text{FR}_8^{\text{6,0}} $ & $ \text{} $ & 6 & 0 & 56 & 4.24264 & 18. & True & False & False & False \\
 $ \text{FR}_9^{\text{6,0}} $ & $ \text{} $ & 6 & 0 & 58 & 4.47214 & 20. & True & True & False & False \\
 $ \text{FR}_{10}^{\text{6,0}} $ & $ \text{} $ & 6 & 0 & 62 & 5.05792 & 25.5826 & True & False & False & True \\
 $ \text{FR}_{11}^{\text{6,0}} $ & $ \text{} $ & 6 & 0 & 65 & 5.32844 & 28.3923 & True & False & True & True \\
 $ \text{FR}_{12}^{\text{6,0}} $ & $ \text{} $ & 6 & 0 & 65 & 5.32844 & 28.3923 & True & False & True & True \\
 $ \text{FR}_{13}^{\text{6,0}} $ & $ \text{} $ & 6 & 0 & 66 & 5.8136 & 33.798 & True & False & False & False \\
 $ \text{FR}_{14}^{\text{6,0}} $ & $ \left.\text{Fib$\times $ExtRep(}D_3\right) $ & 6 & 0 & 70 & 5.79939 & 33.6329 & True & True & False & False \\
 $ \text{FR}_{15}^{\text{6,0}} $ & $ \text{} $ & 6 & 0 & 72 & 6.06459 & 36.7792 & True & False & True & True \\
 $ \text{FR}_{16}^{\text{6,0}} $ & $ \text{PSU}(2)_{10} $ & 6 & 0 & 76 & 6.69213 & 44.7846 & True & False & False & False \\
 $ \text{FR}_{17}^{\text{6,0}} $ & $ \text{} $ & 6 & 0 & 83 & 7.42591 & 55.1442 & True & False & True & True \\
 $ \text{FR}_{18}^{\text{6,0}} $ & $ \text{PSU}(2)_{11} $ & 6 & 0 & 91 & 7.53304 & 56.7468 & True & True & False & False \\
 $ \text{FR}_{19}^{\text{6,0}} $ & $ \text{} $ & 6 & 0 & 101 & 7.94652 & 63.1472 & True & False & False & True \\
 $ \text{FR}_{20}^{\text{6,0}} $ & $ \text{} $ & 6 & 0 & 101 & 7.94652 & 63.1472 & True & False & True & True \\
 $ \text{FR}_1^{\text{6,2}} $ & $D_3$ & 6 & 2 & 36 & 2.44949 & 6. & False & False & False & False \\
 $ \text{FR}_2^{\text{6,2}} $ & $ [\mathbb{Z}_2 \trianglelefteq \mathbb{Z}_4]_{\mathbf{1}|0}^{\text{Id}} $ & 6 & 2 & 40 & 2.82843 & 8. & True & False & False & False \\
 $ \text{FR}_3^{\text{6,2}} $ & $ [\mathbb{Z}_2 \trianglelefteq \mathbb{Z}_2\times  \mathbb{Z}_2]_{\mathbf{3}|0}^{\text{Id}} $ & 6 & 2 & 40 & 2.82843 & 8. & True & False & False & False \\
 $ \text{FR}_4^{\text{6,2}} $ & $ \left.\text{Rep(}\text{Dic}_{12}\right) $ & 6 & 2 & 44 & 3.4641 & 12. & True & False & False & False \\
 $ \text{FR}_5^{\text{6,2}} $ & $ [\mathbb{Z}_2 \trianglelefteq \mathbb{Z}_4]_{\mathbf{1}|1}^{\text{Id}} $ & 6 & 2 & 48 & 4.35066 & 18.9282 & True & False & False & False \\
 $ \text{FR}_6^{\text{6,2}} $ & $ [\mathbb{Z}_2 \trianglelefteq \mathbb{Z}_2\times \
\mathbb{Z}_2]_{\mathbf{3}|1}^{\text{Id}} $ & 6 & 2 & 48 & 4.35066 & 18.9282 & True & False & False & False \\
 $ \text{FR}_7^{\text{6,2}} $ & $ \text{Pseudo SO(5})_2$ & 6 & 2 & 58 & 4.47214 & 20. & True & False & False & False \\
 $ \text{FR}_8^{\text{6,2}} $ & $ \left.\text{HI(}\mathbb{Z}_3\right) $ & 6 & 2 & 63 & 5.97704 & 35.725 & False & False & False & False \\
 $ \text{FR}_9^{\text{6,2}} $ & $ \text{} $ & 6 & 2 & 66 & 5.8136 & 33.798 & True & False & False & False \\
 $ \text{FR}_{10}^{\text{6,2}} $ & $ \text{} $ & 6 & 2 & 72 & 6.06459 & 36.7792 & True & False & True & True \\
 $ \text{FR}_{11}^{\text{6,2}} $ & $ \text{} $ & 6 & 2 & 83 & 7.42591 & 55.1442 & True & False & True & True \\
 $ \text{FR}_1^{\text{6,4}} $ & $ \mathbb{Z}_6$ & 6 & 4 & 36 & 2.44949 & 6. & True & False & False & False \\
 $ \text{FR}_2^{\text{6,4}} $ & $ \text{MR}_6$ & 6 & 4 & 40 & 2.82843 & 8. & True & False & False & False \\
 $ \text{FR}_3^{\text{6,4}} $ & $ \left.\text{TY(}\mathbb{Z}_5\right) $ & 6 & 4 & 40 & 3.16228 & 10. & True & False & False & False \\
 $ \text{FR}_4^{\text{6,4}} $ & $ [\mathbb{Z}_5 \trianglelefteq \mathbb{Z}_5]_{\mathbf{1}|1}^{\text{Id}} $ & 6 & 4 & 41 & 3.57649 & 12.7913 & True & False & False & False \\
 $ \text{FR}_5^{\text{6,4}} $ & $ \text{Fib$\times $}\mathbb{Z}_3$ & 6 & 4 & 45 & 3.29456 & 10.8541 & True & True & False & False \\
 $ \text{FR}_6^{\text{6,4}} $ & $ [\mathbb{Z}_2 \trianglelefteq \mathbb{Z}_4]_{\mathbf{3}|1}^{\text{Id}} $ & 6 & 4 & 48 & 4.35066 & 18.9282 & True & False & False & False \\
 $ \text{FR}_7^{\text{6,4}} $ & $ \text{} $ & 6 & 4 & 54 & 4.52607 & 20.4853 & True & False & False & False \\
 $ \text{FR}_8^{\text{6,4}} $ & $ [I \trianglelefteq \mathbb{Z}_3]_{\mathbf{1}|1}^{\text{Id}} $ & 6 & 4 & 63 & 5.97704 & 35.725 & True & False & False & False \\
 $ \text{FR}_1^{\text{7,0}} $ & $ \text{} $ & 7 & 0 & 64 & 4. & 16. & True & False & False & False \\
 $ \text{FR}_2^{\text{7,0}} $ & $ \text{} $ & 7 & 0 & 68 & 4.59599 & 21.1231 & True & False & True & True \\
 $ \text{FR}_3^{\text{7,0}} $ & $ \text{} $ & 7 & 0 & 71 & 5.21092 & 27.1537 & True & False & False & True \\
 $ \text{FR}_4^{\text{7,0}} $ & $ \text{} $ & 7 & 0 & 72 & 5.37999 & 28.9443 & True & False & False & False \\
 $ \text{FR}_5^{\text{7,0}} $ & $ \text{} $ & 7 & 0 & 74 & 5.4277 & 29.46 & True & False & False & False \\
 $ \text{FR}_6^{\text{7,0}} $ & $ \text{} $ & 7 & 0 & 79 & 4.69042 & 22. & True & False & False & False \\
 $ \text{FR}_7^{\text{7,0}} $ & $ \text{SU}(2)_6$ & 7 & 0 & 84 & 5.22625 & 27.3137 & True & True & False & False \\
 $ \text{FR}_8^{\text{7,0}} $ & $ \text{} $ & 7 & 0 & 85 & 5.2915 & 28. & True & True & False & False \\
 $ \text{FR}_9^{\text{7,0}} $ & $ \text{} $ & 7 & 0 & 89 & 5.86389 & 34.3852 & True & False & False & True \\
 $ \text{FR}_{10}^{\text{7,0}} $ & $ \text{} $ & 7 & 0 & 93 & 6.58132 & 43.3137 & True & False & False & False \\
 $ \text{FR}_{11}^{\text{7,0}} $ & $ \text{} $ & 7 & 0 & 94 & 6.08034 & 36.9706 & True & False & True & True \\
 $ \text{FR}_{12}^{\text{7,0}} $ & $ \text{} $ & 7 & 0 & 99 & 6.48074 & 42. & True & False & True & True \\
 $ \text{FR}_{13}^{\text{7,0}} $ & $ \text{} $ & 7 & 0 & 103 & 7.2753 & 52.93 & True & False & True & True \\
 $ \text{FR}_{14}^{\text{7,0}} $ & $ \text{PSU}(2)_{12} $ & 7 & 0 & 119 & 8.40743 & 70.6848 & True & False & False & False \\
 $ \text{FR}_{15}^{\text{7,0}} $ & $ \text{} $ & 7 & 0 & 129 & 9.03712 & 81.6695 & True & False & False & True \\
 $ \text{FR}_{16}^{\text{7,0}} $ & $ \text{} $ & 7 & 0 & 131 & 9.3324 & 87.0937 & True & False & True & True \\
 $ \text{FR}_{17}^{\text{7,0}} $ & $ \text{PSU}(2)_{13} $ & 7 & 0 & 140 & 9.31401 & 86.7508 & True & True & False & False \\
 $ \text{FR}_{18}^{\text{7,0}} $ & $ \text{} $ & 7 & 0 & 175 & 10.8691 & 118.138 & True & False & False & True \\
 $ \text{FR}_1^{\text{7,2}} $ & $ \text{TY}(D_3) $ & 7 & 2 & 54 & 3.4641 & 12. & False & False & False & False \\
 $ \text{FR}_2^{\text{7,2}} $ & $ [D_3 \trianglelefteq D_3]_{\mathbf{1}|1}^{\text{Id}} $ & 7 & 2 & 55 & 3.87298 & 15. & False & False & False & False \\
 $ \text{FR}_3^{\text{7,2}} $ & $ \text{} $ & 7 & 2 & 64 & 4. & 16. & True & False & False & False \\
 $ \text{FR}_4^{\text{7,2}} $ & $ \text{} $ & 7 & 2 & 64 & 4. & 16. & True & False & False & False \\
 $ \text{FR}_5^{\text{7,2}} $ & $ \text{} $ & 7 & 2 & 68 & 4.59599 & 21.1231 & True & False & True & True \\
 $ \text{FR}_6^{\text{7,2}} $ & $ \text{} $ & 7 & 2 & 71 & 5.21092 & 27.1537 & True & False & False & True \\
 $ \text{FR}_7^{\text{7,2}} $ & $ \text{} $ & 7 & 2 & 71 & 5.21092 & 27.1537 & True & False & False & True \\
 $ \text{FR}_8^{\text{7,2}} $ & $ \text{} $ & 7 & 2 & 72 & 5.37999 & 28.9443 & True & False & False & False \\
 $ \text{FR}_9^{\text{7,2}} $ & $ \text{} $ & 7 & 2 & 72 & 5.37999 & 28.9443 & True & False & False & False \\
 $ \text{FR}_{10}^{\text{7,2}} $ & $ \text{} $ & 7 & 2 & 74 & 5.4277 & 29.46 & True & False & False & False \\
 $ \text{FR}_{11}^{\text{7,2}} $ & $ \text{} $ & 7 & 2 & 84 & 5.22625 & 27.3137 & True & True & False & False \\
 $ \text{FR}_{12}^{\text{7,2}} $ & $ \text{} $ & 7 & 2 & 85 & 5.2915 & 28. & True & False & False & False \\
 $ \text{FR}_{13}^{\text{7,2}} $ & $ \text{} $ & 7 & 2 & 93 & 6.58132 & 43.3137 & True & False & False & False \\
 $ \text{FR}_{14}^{\text{7,2}} $ & $ \text{} $ & 7 & 2 & 103 & 7.2753 & 52.93 & True & False & True & True \\
 $ \text{FR}_{15}^{\text{7,2}} $ & $ \text{} $ & 7 & 2 & 108 & 8.42685 & 71.0118 & False & False & False & False \\
 $ \text{FR}_{16}^{\text{7,2}} $ & $ \text{} $ & 7 & 2 & 129 & 9.03712 & 81.6695 & True & False & False & True \\
 $ \text{FR}_{17}^{\text{7,2}} $ & $ \text{} $ & 7 & 2 & 131 & 9.3324 & 87.0937 & True & False & True & True \\
 $ \text{FR}_1^{\text{7,4}} $ & $ \left.\text{TY(}\mathbb{Z}_2\times \mathbb{Z}_3\right) $ & 7 & 4 & 54 & 3.4641 & 12. & True & False & False & False \\
 $ \text{FR}_2^{\text{7,4}} $ & $ [\mathbb{Z}_6 \trianglelefteq \mathbb{Z}_6]_{\mathbf{1}|1}^{\text{Id}} $ & 7 & 4 & 55 & 3.87298 & 15. & True & False & False & False \\
 $ \text{FR}_3^{\text{7,4}} $ & $ \text{} $ & 7 & 4 & 64 & 4. & 16. & True & False & False & False \\
 $ \text{FR}_4^{\text{7,4}} $ & $ \text{} $ & 7 & 4 & 71 & 5.21092 & 27.1537 & True & False & False & True \\
 $ \text{FR}_5^{\text{7,4}} $ & $ \text{} $ & 7 & 4 & 72 & 5.37999 & 28.9443 & True & False & False & False \\
 $ \text{FR}_6^{\text{7,4}} $ & $ \text{} $ & 7 & 4 & 100 & 7.56541 & 57.2354 & True & False & False & False \\
 $ \text{FR}_7^{\text{7,4}} $ & $ \text{} $ & 7 & 4 & 108 & 8.42685 & 71.0118 & True & False & False & False \\
 $ \text{FR}_1^{\text{7,6}} $ & $ \mathbb{Z}_7$ & 7 & 6 & 49 & 2.64575 & 7. & True & True & False & False \\
 $ \text{FR}_1^{\text{8,0}} $ & $ \mathbb{Z}_2\times \mathbb{Z}_2\times \mathbb{Z}_2$ & 8 & 0 & 64 & 2.82843 & 8. & True & True & False & False \\
 $ \text{FR}_2^{\text{8,0}} $ & $ \text{Fib$\times $}\mathbb{Z}_2\times \mathbb{Z}_2$ & 8 & 0 & 80 & 3.80423 & 14.4721 & True & True & False & False \\
 $ \text{FR}_3^{\text{8,0}} $ & $ \text{Rep(}D_5\text{)$\times $}\mathbb{Z}_2$ & 8 & 0 & 88 & 4.47214 & 20. & True & False & False & False \\
 $ \text{FR}_4^{\text{8,0}} $ & $ \text{} $ & 8 & 0 & 92 & 4.89898 & 24. & True & False & False & False \\
 $ \text{FR}_5^{\text{8,0}} $ & $ \text{PSU}(2)_6\times \mathbb{Z}_2$ & 8 & 0 & 96 & 5.22625 & 27.3137 & True & False & False & False \\
 $ \text{FR}_6^{\text{8,0}} $ & $ \text{} $ & 8 & 0 & 96 & 5.47723 & 30. & True & False & False & True \\
 $ \text{FR}_7^{\text{8,0}} $ & $ \text{Fib$\times $Fib$\times $}\mathbb{Z}_2$ & 8 & 0 & 100 & 5.11667 & 26.1803 & True & True & False & False \\
 $ \text{FR}_8^{\text{8,0}} $ & $ \text{} $ & 8 & 0 & 100 & 6.21152 & 38.583 & True & False & False & False \\
 $ \text{FR}_9^{\text{8,0}} $ & $ \text{} $ & 8 & 0 & 106 & 5.09902 & 26. & True & False & False & False \\
 $ \text{FR}_{10}^{\text{8,0}} $ & $ \text{} $ & 8 & 0 & 108 & 6.51602 & 42.4585 & True & False & False & True \\
 $ \text{FR}_{11}^{\text{8,0}} $ & $ \left.\text{Fib$\times $}\text{Rep}(D_5\right) $ & 8 & 0 & 110 & 6.01501 & 36.1803 & True & False & False & False \\
 $ \text{FR}_{12}^{\text{8,0}} $ & $ \text{} $ & 8 & 0 & 112 & 6.90169 & 47.6333 & True & False & False & True \\
 $ \text{FR}_{13}^{\text{8,0}} $ & $ \text{} $ & 8 & 0 & 116 & 6. & 36. & True & True & False & False \\
 $ \text{FR}_{14}^{\text{8,0}} $ & $ \text{} $ & 8 & 0 & 116 & 6. & 36. & True & True & False & False \\
 $ \text{FR}_{15}^{\text{8,0}} $ & $ \text{SU}(2)_7$ & 8 & 0 & 120 & 6.20233 & 38.4688 & True & True & False & False \\
 $ \text{FR}_{16}^{\text{8,0}} $ & $ \text{Fib$\times $}\text{PSU}(2)_6$ & 8 & 0 & 120 & 7.02929 & 49.411 & True & False & False & False \\
 $ \text{FR}_{17}^{\text{8,0}} $ & $ \text{} $ & 8 & 0 & 120 & 6.56375 & 43.0828 & True & False & False & True \\
 $ \text{FR}_{18}^{\text{8,0}} $ & $ \text{} $ & 8 & 0 & 120 & 6.56375 & 43.0828 & True & False & False & True \\
 $ \text{FR}_{19}^{\text{8,0}} $ & $ \text{} $ & 8 & 0 & 124 & 8.28637 & 68.6639 & True & False & False & False \\
 $ \text{FR}_{20}^{\text{8,0}} $ & $ \text{} $ & 8 & 0 & 124 & 7.25597 & 52.6491 & True & False & False & False \\
 $ \text{FR}_{21}^{\text{8,0}} $ & $ \text{} $ & 8 & 0 & 124 & 7.25597 & 52.6491 & True & False & False & False \\
 $ \text{FR}_{22}^{\text{8,0}} $ & $ \text{Fib$\times $Fib$\times $Fib} $ & 8 & 0 & 125 & 6.88191 & 47.3607 & True & True & False & False \\
 $ \text{FR}_{23}^{\text{8,0}} $ & $ \left.\text{HI}(\mathbb{Z}_2\times \mathbb{Z}_2\right) $ & 8 & 0 & 128 & 8.705 & 75.7771 & True & False & False & False \\
 $ \text{FR}_{24}^{\text{8,0}} $ & $ \text{} $ & 8 & 0 & 130 & 6.93315 & 48.0685 & True & False & True & True \\
 $ \text{FR}_{25}^{\text{8,0}} $ & $ \text{} $ & 8 & 0 & 142 & 7.62344 & 58.1168 & True & False & False & True \\
 $ \text{FR}_{26}^{\text{8,0}} $ & $ \text{} $ & 8 & 0 & 142 & 7.62344 & 58.1168 & True & False & False & True \\
 $ \text{FR}_{27}^{\text{8,0}} $ & $ \text{Fib$\times $}\text{PSU}(2)_7$ & 8 & 0 & 150 & 8.34211 & 69.5908 & True & True & False & False \\
 $ \text{FR}_{28}^{\text{8,0}} $ & $ \text{} $ & 8 & 0 & 151 & 8.48528 & 72. & True & False & False & False \\
 $ \text{FR}_{29}^{\text{8,0}} $ & $ \text{} $ & 8 & 0 & 151 & 8.48528 & 72. & True & False & False & False \\
 $ \text{FR}_{30}^{\text{8,0}} $ & $ \text{} $ & 8 & 0 & 160 & 8.84102 & 78.1637 & True & False & True & True \\
 $ \text{FR}_{31}^{\text{8,0}} $ & $ \text{PSU}(2)_{14} $ & 8 & 0 & 176 & 10.2517 & 105.097 & True & False & False & False \\
 $ \text{FR}_{32}^{\text{8,0}} $ & $ \text{} $ & 8 & 0 & 187 & 11.2482 & 126.522 & True & False & False & False \\
 $ \text{FR}_{33}^{\text{8,0}} $ & $ \text{} $ & 8 & 0 & 187 & 11.2482 & 126.522 & True & False & False & False \\
 $ \text{FR}_{34}^{\text{8,0}} $ & $ \text{} $ & 8 & 0 & 194 & 11.3212 & 128.169 & True & False & True & True \\
 $ \text{FR}_{35}^{\text{8,0}} $ & $ \text{} $ & 8 & 0 & 196 & 11.0713 & 122.573 & True & False & False & True \\
 $ \text{FR}_{36}^{\text{8,0}} $ & $ \text{PSU}(2)_{15} $ & 8 & 0 & 204 & 11.2194 & 125.874 & True & True & False & False \\
 $ \text{FR}_{37}^{\text{8,0}} $ & $ \text{} $ & 8 & 0 & 208 & 11.8569 & 140.586 & True & False & False & True \\
 $ \text{FR}_{38}^{\text{8,0}} $ & $ \text{} $ & 8 & 0 & 281 & 14.1819 & 201.126 & True & False & False & True \\
 $ \text{FR}_1^{\text{8,2}} $ & $D_4$ & 8 & 2 & 64 & 2.82843 & 8. & False & False & False & False \\
 $ \text{FR}_2^{\text{8,2}} $ & $ [\mathbb{Z}_3 \trianglelefteq D_3]_{\mathbf{1}|0}^{\text{Id}} $ & 8 & 2 & 72 & 3.4641 & 12. & False & False & False & False \\
 $ \text{FR}_3^{\text{8,2}} $ & $ \text{} $ & 8 & 2 & 76 & 4.07499 & 16.6056 & False & False & False & False \\
 $ \text{FR}_4^{\text{8,2}} $ & $ [\mathbb{Z}_3 \trianglelefteq D_3]_{\mathbf{1}|1}^{\text{Id}} $ & 8 & 2 & 80 & 4.89898 & 24. & False & False & False & False \\
 $ \text{FR}_5^{\text{8,2}} $ & $ \text{} $ & 8 & 2 & 88 & 4.47214 & 20. & True & False & False & False \\
 $ \text{FR}_6^{\text{8,2}} $ & $ \text{} $ & 8 & 2 & 88 & 4.47214 & 20. & False & False & False & False \\
 $ \text{FR}_7^{\text{8,2}} $ & $ \text{} $ & 8 & 2 & 92 & 4.89898 & 24. & True & False & False & False \\
 $ \text{FR}_8^{\text{8,2}} $ & $ \text{} $ & 8 & 2 & 96 & 5.22625 & 27.3137 & False & False & False & False \\
 $ \text{FR}_9^{\text{8,2}} $ & $ \text{} $ & 8 & 2 & 96 & 5.22625 & 27.3137 & False & False & False & False \\
 $ \text{FR}_{10}^{\text{8,2}} $ & $ \text{} $ & 8 & 2 & 96 & 5.22625 & 27.3137 & False & False & False & False \\
 $ \text{FR}_{11}^{\text{8,2}} $ & $ \text{} $ & 8 & 2 & 100 & 5.11667 & 26.1803 & False & False & False & False \\
 $ \text{FR}_{12}^{\text{8,2}} $ & $ \text{} $ & 8 & 2 & 100 & 6.21152 & 38.583 & True & False & False & False \\
 $ \text{FR}_{13}^{\text{8,2}} $ & $ \text{} $ & 8 & 2 & 108 & 6.51602 & 42.4585 & True & False & False & True \\
 $ \text{FR}_{14}^{\text{8,2}} $ & $ \text{} $ & 8 & 2 & 116 & 6. & 36. & True & False & False & False \\
 $ \text{FR}_{15}^{\text{8,2}} $ & $ \text{} $ & 8 & 2 & 116 & 6. & 36. & True & False & False & False \\
 $ \text{FR}_{16}^{\text{8,2}} $ & $ \text{} $ & 8 & 2 & 124 & 8.28637 & 68.6639 & True & False & False & False \\
 $ \text{FR}_{17}^{\text{8,2}} $ & $ \text{} $ & 8 & 2 & 124 & 7.25597 & 52.6491 & True & False & False & False \\
 $ \text{FR}_{18}^{\text{8,2}} $ & $ \text{} $ & 8 & 2 & 124 & 7.25597 & 52.6491 & True & False & False & False \\
 $ \text{FR}_{19}^{\text{8,2}} $ & $ \left.\text{HI}(\mathbb{Z}_4\right) $ & 8 & 2 & 128 & 8.705 & 75.7771 & False & False & False & False \\
 $ \text{FR}_{20}^{\text{8,2}} $ & $ [I \trianglelefteq \mathbb{Z}_2\times \mathbb{Z}_2]_{\mathbf{1}|1}^{( \mathbf{3} \  \mathbf{4} )} $ & 8 & 2 & 128 & 8.705 & 75.7771 & False & False & False & False \\
 $ \text{FR}_{21}^{\text{8,2}} $ & $ \text{} $ & 8 & 2 & 130 & 6.93315 & 48.0685 & True & False & True & True \\
 $ \text{FR}_{22}^{\text{8,2}} $ & $ \text{} $ & 8 & 2 & 142 & 7.62344 & 58.1168 & True & False & False & True \\
 $ \text{FR}_{23}^{\text{8,2}} $ & $ \text{} $ & 8 & 2 & 142 & 7.62344 & 58.1168 & True & False & False & True \\
 $ \text{FR}_{24}^{\text{8,2}} $ & $ \text{} $ & 8 & 2 & 151 & 8.48528 & 72. & True & False & False & False \\
 $ \text{FR}_{25}^{\text{8,2}} $ & $ \text{} $ & 8 & 2 & 151 & 8.48528 & 72. & True & False & False & False \\
 $ \text{FR}_{26}^{\text{8,2}} $ & $ \text{} $ & 8 & 2 & 187 & 11.2482 & 126.522 & True & False & False & False \\
 $ \text{FR}_{27}^{\text{8,2}} $ & $ \text{} $ & 8 & 2 & 187 & 11.2482 & 126.522 & True & False & False & False \\
 $ \text{FR}_{28}^{\text{8,2}} $ & $ \text{} $ & 8 & 2 & 194 & 11.3212 & 128.169 & True & False & True & True \\
 $ \text{FR}_{29}^{\text{8,2}} $ & $ \text{} $ & 8 & 2 & 196 & 11.0713 & 122.573 & False & False & True & False \\
 $ \text{FR}_{30}^{\text{8,2}} $ & $ \text{} $ & 8 & 2 & 208 & 11.8569 & 140.586 & False & False & False & False \\
 $ \text{FR}_1^{\text{8,4}} $ & $ \mathbb{Z}_2\times \mathbb{Z}_4$ & 8 & 4 & 64 & 2.82843 & 8. & True & True & False & False \\
 $ \text{FR}_2^{\text{8,4}} $ & $ [\mathbb{Z}_3 \trianglelefteq D_3]_{\mathbf{2}|0}^{\text{Id}} $ & 8 & 4 & 72 & 3.4641 & 12. & False & False & False & False \\
 $ \text{FR}_3^{\text{8,4}} $ & $ \mathbb{Z}_2\text{$\times $Potts} $ & 8 & 4 & 72 & 3.4641 & 12. & True & False & False & False \\
 $ \text{FR}_4^{\text{8,4}} $ & $ \left.\mathbb{Z}_2\text{$\times $Fib(}\mathbb{Z}_3\right) $ & 8 & 4 & 76 & 4.07499 & 16.6056 & True & False & False & False \\
 $ \text{FR}_5^{\text{8,4}} $ & $ [\mathbb{Z}_3 \trianglelefteq D_3]_{\mathbf{2}|1}^{\text{Id}} $ & 8 & 4 & 80 & 4.89898 & 24. & False & False & False & False \\
 $ \text{FR}_6^{\text{8,4}} $ & $ [\mathbb{Z}_3 \trianglelefteq \mathbb{Z}_6]_{\mathbf{1}|1}^{\text{Id}} $ & 8 & 4 & 80 & 4.89898 & 24. & True & False & False & False \\
 $ \text{FR}_7^{\text{8,4}} $ & $ \text{Fib$\times $}\mathbb{Z}_4$ & 8 & 4 & 80 & 3.80423 & 14.4721 & True & False & False & False \\
 $ \text{FR}_8^{\text{8,4}} $ & $ \text{} $ & 8 & 4 & 88 & 4.47214 & 20. & False & False & False & False \\
 $ \text{FR}_9^{\text{8,4}} $ & $ \text{Fib$\times $Potts} $ & 8 & 4 & 90 & 4.65921 & 21.7082 & True & False & False & False \\
 $ \text{FR}_{10}^{\text{8,4}} $ & $ \left.\text{Fib$\times $Fib(}\mathbb{Z}_3\right) $ & 8 & 4 & 95 & 5.48085 & 30.0397 & True & False & False & False \\
 $ \text{FR}_{11}^{\text{8,4}} $ & $ \text{} $ & 8 & 4 & 96 & 5.22625 & 27.3137 & True & False & False & False \\
 $ \text{FR}_{12}^{\text{8,4}} $ & $ \text{} $ & 8 & 4 & 96 & 5.22625 & 27.3137 & True & False & False & False \\
 $ \text{FR}_{13}^{\text{8,4}} $ & $ \left.\mathbb{Z}_2\text{$\times $(Peudo }\text{PSU}(2)_6\right) $ & 8 & 4 & 96 & 5.22625 & 27.3137 & True & False & False & False \\
 $ \text{FR}_{14}^{\text{8,4}} $ & $ \text{} $ & 8 & 4 & 112 & 6.90169 & 47.6333 & True & False & False & False \\
 $ \text{FR}_{15}^{\text{8,4}} $ & $ \left.\text{Fib$\times $(Peudo }\text{PSU}(2)_6\right) $ & 8 & 4 & 120 & 7.02929 & 49.411 & True & False & False & False \\
 $ \text{FR}_{16}^{\text{8,4}} $ & $ [I \trianglelefteq \mathbb{Z}_4]_{\mathbf{1}|1}^{\text{Id}} $ & 8 & 4 & 128 & 8.705 & 75.7771 & True & False & False & False \\
 $ \text{FR}_{17}^{\text{8,4}} $ & $ [I \trianglelefteq \mathbb{Z}_2\times \mathbb{Z}_2]_{\mathbf{2}|1}^{\text{Id}} $ & 8 & 4 & 128 & 8.705 & 75.7771 & True & False & False & False \\
 $ \text{FR}_{18}^{\text{8,4}} $ & $ \text{} $ & 8 & 4 & 160 & 8.84102 & 78.1637 & True & False & True & False \\
 $ \text{FR}_1^{\text{8,6}} $ & $ \text{Q} $ & 8 & 6 & 64 & 2.82843 & 8. & False & False & False & False \\
 $ \text{FR}_2^{\text{8,6}} $ & $ \mathbb{Z}_8$ & 8 & 6 & 64 & 2.82843 & 8. & True & True & False & False \\
 $ \text{FR}_3^{\text{8,6}} $ & $ \left.\text{TY(}\mathbb{Z}_7\right) $ & 8 & 6 & 70 & 3.74166 & 14. & True & False & False & False \\
 $ \text{FR}_4^{\text{8,6}} $ & $ [\mathbb{Z}_7 \trianglelefteq \mathbb{Z}_7]_{\mathbf{1}|1}^{\text{Id}} $ & 8 & 6 & 71 & 4.14639 & 17.1926 & True & False & False & False \\
 $ \text{FR}_5^{\text{8,6}} $ & $ [\mathbb{Z}_3 \trianglelefteq \mathbb{Z}_6]_{\mathbf{2}|0}^{\text{Id}} $ & 8 & 6 & 72 & 3.4641 & 12. & True & False & False & False \\
 $ \text{FR}_6^{\text{8,6}} $ & $ [\mathbb{Z}_3 \trianglelefteq \mathbb{Z}_6]_{\mathbf{2}|1}^{\text{Id}} $ & 8 & 6 & 80 & 4.89898 & 24. & True & False & False & False \\
 $ \text{FR}_7^{\text{8,6}} $ & $ \text{} $ & 8 & 6 & 96 & 5.22625 & 27.3137 & False & False & False & False \\
 $ \text{FR}_8^{\text{8,6}} $ & $ \text{} $ & 8 & 6 & 96 & 5.22625 & 27.3137 & True & False & False & False \\
 $ \text{FR}_9^{\text{8,6}} $ & $ [I \trianglelefteq \mathbb{Z}_4]_{\mathbf{2}|1}^{( \mathbf{3} \  \mathbf{4} )} $ & 8 & 6 & 128 & 8.705 & 75.7771 & False & False & False & False \\
 $ \text{FR}_{10}^{\text{8,6}} $ & $ [I \trianglelefteq \mathbb{Z}_4]_{\mathbf{3}|1}^{\text{Id}} $ & 8 & 6 & 128 & 8.705 & 75.7771 & True & False & False & False \\
 $ \text{FR}_1^{\text{9,0}} $ & $ \text{TY}\left((\mathbb{Z}_2)^{\times 3}\right) $ & 9 & 0 & 88 & 4. & 16. & True & False & False & False \\
 $ \text{FR}_2^{\text{9,0}} $ & $ [\left((\mathbb{Z}_2)^{\times 3}\right) \trianglelefteq \left((\mathbb{Z}_2)^{\times 3}\right)]_{\mathbf{1}|1}^{\text{Id}} $ & 9 & 0 & 89 & 4.4014 & 19.3723 & True & False & False & False \\
 $ \text{FR}_3^{\text{9,0}} $ & $ \text{Ising$\times $Ising} $ & 9 & 0 & 100 & 4. & 16. & True & True & False & False \\
 $ \text{FR}_4^{\text{9,0}} $ & $ \left.\text{Ising$\times $Rep(}D_3\right) $ & 9 & 0 & 110 & 4.89898 & 24. & True & False & False & False \\
 $ \text{FR}_5^{\text{9,0}} $ & $ \text{} $ & 9 & 0 & 116 & 4.89898 & 24. & True & False & False & False \\
 $ \text{FR}_6^{\text{9,0}} $ & $ \left.\text{Rep(}D_3\text{)$\times $Rep(}D_3\right) $ & 9 & 0 & 121 & 6. & 36. & True & False & False & False \\
 $ \text{FR}_7^{\text{9,0}} $ & $ \text{} $ & 9 & 0 & 124 & 5.65685 & 32. & True & False & False & False \\
 $ \text{FR}_8^{\text{9,0}} $ & $ \text{} $ & 9 & 0 & 128 & 6.22451 & 38.7446 & True & False & False & True \\
 $ \text{FR}_9^{\text{9,0}} $ & $ \text{} $ & 9 & 0 & 132 & 6.15276 & 37.8564 & True & False & False & False \\
 $ \text{FR}_{10}^{\text{9,0}} $ & $ \text{} $ & 9 & 0 & 132 & 6.9282 & 48. & True & False & False & False \\
 $ \text{FR}_{11}^{\text{9,0}} $ & $ \text{} $ & 9 & 0 & 132 & 6.15276 & 37.8564 & True & False & True & True \\
 $ \text{FR}_{12}^{\text{9,0}} $ & $ \text{} $ & 9 & 0 & 137 & 5.47723 & 30. & True & False & False & False \\
 $ \text{FR}_{13}^{\text{9,0}} $ & $ \text{} $ & 9 & 0 & 140 & 7.03925 & 49.551 & True & False & True & True \\
 $ \text{FR}_{14}^{\text{9,0}} $ & $ \text{Ising$\times $}\text{PSU}(2)_5$ & 9 & 0 & 140 & 6.09783 & 37.1836 & True & True & False & False \\
 $ \text{FR}_{15}^{\text{9,0}} $ & $ \text{} $ & 9 & 0 & 143 & 7.74597 & 60. & True & False & False & False \\
 $ \text{FR}_{16}^{\text{9,0}} $ & $ \text{} $ & 9 & 0 & 148 & 7.66411 & 58.7386 & True & False & False & True \\
 $ \text{FR}_{17}^{\text{9,0}} $ & $ \text{} $ & 9 & 0 & 151 & 7.63029 & 58.2213 & True & False & True & True \\
 $ \text{FR}_{18}^{\text{9,0}} $ & $ \text{} $ & 9 & 0 & 151 & 7.63029 & 58.2213 & True & False & False & True \\
 $ \text{FR}_{19}^{\text{9,0}} $ & $ \text{} $ & 9 & 0 & 151 & 6.63325 & 44. & True & True & False & False \\
 $ \text{FR}_{20}^{\text{9,0}} $ & $ \text{Rep(}D_3\text{)$\times $}\text{PSU}(2)_5$ & 9 & 0 & 154 & 7.46829 & 55.7754 & True & False & False & False \\
 $ \text{FR}_{21}^{\text{9,0}} $ & $ \text{} $ & 9 & 0 & 155 & 7.19084 & 51.7082 & True & False & False & True \\
 $ \text{FR}_{22}^{\text{9,0}} $ & $ \text{} $ & 9 & 0 & 156 & 8.48528 & 72. & True & False & False & False \\
 $ \text{FR}_{23}^{\text{9,0}} $ & $ \text{} $ & 9 & 0 & 156 & 8.62364 & 74.3672 & True & False & False & False \\
 $ \text{FR}_{24}^{\text{9,0}} $ & $ \text{} $ & 9 & 0 & 159 & 7.86488 & 61.8564 & True & False & False & False \\
 $ \text{FR}_{25}^{\text{9,0}} $ & $ \text{} $ & 9 & 0 & 163 & 6.9282 & 48. & True & False & True & True \\
 $ \text{FR}_{26}^{\text{9,0}} $ & $ \text{} $ & 9 & 0 & 164 & 9.33766 & 87.1918 & True & False & False & False \\
 $ \text{FR}_{27}^{\text{9,0}} $ & $ \text{SU}(2)_8$ & 9 & 0 & 165 & 7.23607 & 52.3607 & True & True & False & False \\
 $ \text{FR}_{28}^{\text{9,0}} $ & $ \text{} $ & 9 & 0 & 179 & 7.74597 & 60. & True & False & True & True \\
 $ \text{FR}_{29}^{\text{9,0}} $ & $ \text{} $ & 9 & 0 & 185 & 10.4077 & 108.321 & True & False & False & False \\
 $ \text{FR}_{30}^{\text{9,0}} $ & $ \text{} $ & 9 & 0 & 185 & 8.37915 & 70.2101 & True & False & True & True \\
 $ \text{FR}_{31}^{\text{9,0}} $ & $ \text{} $ & 9 & 0 & 191 & 9.38658 & 88.108 & True & False & True & True \\
 $ \text{FR}_{32}^{\text{9,0}} $ & $ \text{} $ & 9 & 0 & 195 & 8.78442 & 77.166 & True & False & True & True \\
 $ \text{FR}_{33}^{\text{9,0}} $ & $ \text{} $ & 9 & 0 & 195 & 8.78442 & 77.166 & True & False & True & True \\
 $ \text{FR}_{34}^{\text{9,0}} $ & $ \text{PSU}(2)_5\times \text{PSU}(2)_5$ & 9 & 0 & 196 & 9.2959 & 86.4137 & True & True & False & False \\
 $ \text{FR}_{35}^{\text{9,0}} $ & $ \text{} $ & 9 & 0 & 200 & 11.6179 & 134.976 & True & False & False & False \\
 $ \text{FR}_{36}^{\text{9,0}} $ & $ \text{} $ & 9 & 0 & 205 & 9.8961 & 97.9329 & True & False & True & True \\
 $ \text{FR}_{37}^{\text{9,0}} $ & $ \text{} $ & 9 & 0 & 211 & 10.0233 & 100.467 & True & False & False & True \\
 $ \text{FR}_{38}^{\text{9,0}} $ & $ \text{} $ & 9 & 0 & 212 & 10.4398 & 108.99 & True & False & False & False \\
 $ \text{FR}_{39}^{\text{9,0}} $ & $ \text{} $ & 9 & 0 & 218 & 10.5315 & 110.912 & True & False & True & True \\
 $ \text{FR}_{40}^{\text{9,0}} $ & $ \text{} $ & 9 & 0 & 227 & 11.4279 & 130.596 & True & False & False & False \\
 $ \text{FR}_{41}^{\text{9,0}} $ & $ \text{PSU}(2)_{16} $ & 9 & 0 & 249 & 12.2162 & 149.235 & True & False & False & False \\
 $ \text{FR}_{42}^{\text{9,0}} $ & $ \text{} $ & 9 & 0 & 256 & 11.7082 & 137.082 & True & False & True & True \\
 $ \text{FR}_{43}^{\text{9,0}} $ & $ \text{} $ & 9 & 0 & 274 & 13.401 & 179.586 & True & False & True & True \\
 $ \text{FR}_{44}^{\text{9,0}} $ & $ \text{PSU}(2)_{17} $ & 9 & 0 & 285 & 13.2413 & 175.333 & True & True & False & False \\
 $ \text{FR}_{45}^{\text{9,0}} $ & $ \text{} $ & 9 & 0 & 319 & 15.0837 & 227.519 & True & False & False & True \\
 $ \text{FR}_{46}^{\text{9,0}} $ & $ \text{} $ & 9 & 0 & 425 & 17.8357 & 318.114 & True & False & False & True \\
 $ \text{FR}_1^{\text{9,2}} $ & $ \text{TY}(D_4) $ & 9 & 2 & 88 & 4. & 16. & False & False & False & False \\
 $ \text{FR}_2^{\text{9,2}} $ & $ [D_4 \trianglelefteq D_4]_{\mathbf{1}|1}^{\text{Id}} $ & 9 & 2 & 89 & 4.4014 & 19.3723 & False & False & False & False \\
 $ \text{FR}_3^{\text{9,2}} $ & $ \text{} $ & 9 & 2 & 100 & 4. & 16. & True & False & False & False \\
 $ \text{FR}_4^{\text{9,2}} $ & $ \text{} $ & 9 & 2 & 108 & 5.44702 & 29.67 & False & False & True & False \\
 $ \text{FR}_5^{\text{9,2}} $ & $ \text{} $ & 9 & 2 & 110 & 4.89898 & 24. & True & False & False & False \\
 $ \text{FR}_6^{\text{9,2}} $ & $ \text{} $ & 9 & 2 & 116 & 4.89898 & 24. & True & False & False & False \\
 $ \text{FR}_7^{\text{9,2}} $ & $ \text{} $ & 9 & 2 & 116 & 4.89898 & 24. & False & False & False & False \\
 $ \text{FR}_8^{\text{9,2}} $ & $ \text{} $ & 9 & 2 & 116 & 4.89898 & 24. & True & False & False & False \\
 $ \text{FR}_9^{\text{9,2}} $ & $ \text{} $ & 9 & 2 & 124 & 5.65685 & 32. & True & False & False & False \\
 $ \text{FR}_{10}^{\text{9,2}} $ & $ \text{} $ & 9 & 2 & 124 & 5.65685 & 32. & True & False & False & False \\
 $ \text{FR}_{11}^{\text{9,2}} $ & $ \text{} $ & 9 & 2 & 128 & 6.22451 & 38.7446 & True & False & False & True \\
 $ \text{FR}_{12}^{\text{9,2}} $ & $ \text{} $ & 9 & 2 & 132 & 6.15276 & 37.8564 & True & False & False & False \\
 $ \text{FR}_{13}^{\text{9,2}} $ & $ \text{} $ & 9 & 2 & 132 & 6.15276 & 37.8564 & True & False & False & False \\
 $ \text{FR}_{14}^{\text{9,2}} $ & $ \text{} $ & 9 & 2 & 132 & 6.9282 & 48. & True & False & False & False \\
 $ \text{FR}_{15}^{\text{9,2}} $ & $ \text{} $ & 9 & 2 & 132 & 6.9282 & 48. & True & False & False & False \\
 $ \text{FR}_{16}^{\text{9,2}} $ & $ \text{} $ & 9 & 2 & 132 & 6.15276 & 37.8564 & True & False & True & True \\
 $ \text{FR}_{17}^{\text{9,2}} $ & $ \text{} $ & 9 & 2 & 140 & 7.03925 & 49.551 & True & False & True & True \\
 $ \text{FR}_{18}^{\text{9,2}} $ & $ \text{} $ & 9 & 2 & 140 & 7.03925 & 49.551 & True & False & True & True \\
 $ \text{FR}_{19}^{\text{9,2}} $ & $ \text{} $ & 9 & 2 & 143 & 7.74597 & 60. & True & False & False & False \\
 $ \text{FR}_{20}^{\text{9,2}} $ & $ \text{} $ & 9 & 2 & 143 & 7.74597 & 60. & True & False & False & False \\
 $ \text{FR}_{21}^{\text{9,2}} $ & $ \text{} $ & 9 & 2 & 148 & 7.66411 & 58.7386 & True & False & False & True \\
 $ \text{FR}_{22}^{\text{9,2}} $ & $ \text{} $ & 9 & 2 & 148 & 7.66411 & 58.7386 & False & False & False & False \\
 $ \text{FR}_{23}^{\text{9,2}} $ & $ \text{} $ & 9 & 2 & 148 & 7.66411 & 58.7386 & True & False & False & True \\
 $ \text{FR}_{24}^{\text{9,2}} $ & $ \text{} $ & 9 & 2 & 151 & 7.63029 & 58.2213 & True & False & True & True \\
 $ \text{FR}_{25}^{\text{9,2}} $ & $ \text{} $ & 9 & 2 & 151 & 7.63029 & 58.2213 & True & False & False & True \\
 $ \text{FR}_{26}^{\text{9,2}} $ & $ \text{} $ & 9 & 2 & 151 & 6.63325 & 44. & True & False & False & False \\
 $ \text{FR}_{27}^{\text{9,2}} $ & $ \text{} $ & 9 & 2 & 156 & 8.62364 & 74.3672 & True & False & False & False \\
 $ \text{FR}_{28}^{\text{9,2}} $ & $ \text{} $ & 9 & 2 & 156 & 8.62364 & 74.3672 & True & False & False & False \\
 $ \text{FR}_{29}^{\text{9,2}} $ & $ \text{} $ & 9 & 2 & 159 & 7.86488 & 61.8564 & True & False & False & False \\
 $ \text{FR}_{30}^{\text{9,2}} $ & $ \text{} $ & 9 & 2 & 163 & 6.9282 & 48. & True & False & True & True \\
 $ \text{FR}_{31}^{\text{9,2}} $ & $ \text{} $ & 9 & 2 & 164 & 9.33766 & 87.1918 & True & False & False & False \\
 $ \text{FR}_{32}^{\text{9,2}} $ & $ \text{} $ & 9 & 2 & 179 & 7.74597 & 60. & False & False & True & False \\
 $ \text{FR}_{33}^{\text{9,2}} $ & $ \text{} $ & 9 & 2 & 179 & 7.74597 & 60. & True & False & True & True \\
 $ \text{FR}_{34}^{\text{9,2}} $ & $ \text{} $ & 9 & 2 & 195 & 8.78442 & 77.166 & True & False & True & True \\
 $ \text{FR}_{35}^{\text{9,2}} $ & $ \text{} $ & 9 & 2 & 195 & 8.78442 & 77.166 & True & False & True & True \\
 $ \text{FR}_{36}^{\text{9,2}} $ & $ \text{} $ & 9 & 2 & 200 & 11.6179 & 134.976 & False & False & False & False \\
 $ \text{FR}_{37}^{\text{9,2}} $ & $ \text{} $ & 9 & 2 & 200 & 11.6179 & 134.976 & False & False & False & False \\
 $ \text{FR}_{38}^{\text{9,2}} $ & $ \text{} $ & 9 & 2 & 211 & 10.0233 & 100.467 & False & False & False & False \\
 $ \text{FR}_{39}^{\text{9,2}} $ & $ \text{} $ & 9 & 2 & 211 & 10.0233 & 100.467 & True & False & False & True \\
 $ \text{FR}_{40}^{\text{9,2}} $ & $ \text{} $ & 9 & 2 & 227 & 11.4279 & 130.596 & True & False & False & False \\
 $ \text{FR}_{41}^{\text{9,2}} $ & $ \text{} $ & 9 & 2 & 256 & 11.7082 & 137.082 & False & False & True & False \\
 $ \text{FR}_{42}^{\text{9,2}} $ & $ \text{} $ & 9 & 2 & 256 & 11.7082 & 137.082 & True & False & True & True \\
 $ \text{FR}_{43}^{\text{9,2}} $ & $ \text{} $ & 9 & 2 & 274 & 13.401 & 179.586 & True & False & True & True \\
 $ \text{FR}_{44}^{\text{9,2}} $ & $ \text{} $ & 9 & 2 & 319 & 15.0837 & 227.519 & False & False & False & False \\
 $ \text{FR}_{45}^{\text{9,2}} $ & $ \text{} $ & 9 & 2 & 319 & 15.0837 & 227.519 & True & False & False & True \\
 $ \text{FR}_1^{\text{9,4}} $ & $ \left.\text{TY(}\mathbb{Z}_2\times \mathbb{Z}_4\right) $ & 9 & 4 & 88 & 4. & 16. & True & False & False & False \\
 $ \text{FR}_2^{\text{9,4}} $ & $ [\mathbb{Z}_2\times \mathbb{Z}_4 \trianglelefteq \mathbb{Z}_2\times \mathbb{Z}_4]_{\mathbf{1}|1}^{\text{Id}} $ & 9 & 4 & 89 & 4.4014 & 19.3723 & True & False & False & False \\
 $ \text{FR}_3^{\text{9,4}} $ & $ [\mathbb{Z}_2 \trianglelefteq \mathbb{Z}_6]_{\mathbf{1}|0}^{( \mathbf{2} \  \mathbf{3} )} $ & 9 & 4 & 90 & 3.4641 & 12. & False & False & False & False \\
 $ \text{FR}_4^{\text{9,4}} $ & $ \text{} $ & 9 & 4 & 100 & 4. & 16. & False & False & False & False \\
 $ \text{FR}_5^{\text{9,4}} $ & $ \text{} $ & 9 & 4 & 100 & 4. & 16. & True & False & False & False \\
 $ \text{FR}_6^{\text{9,4}} $ & $ \text{} $ & 9 & 4 & 108 & 5.44702 & 29.67 & False & False & True & False \\
 $ \text{FR}_7^{\text{9,4}} $ & $ \text{} $ & 9 & 4 & 108 & 5.44702 & 29.67 & True & False & True & False \\
 $ \text{FR}_8^{\text{9,4}} $ & $ \text{} $ & 9 & 4 & 116 & 4.89898 & 24. & False & False & False & False \\
 $ \text{FR}_9^{\text{9,4}} $ & $ \text{} $ & 9 & 4 & 116 & 4.89898 & 24. & False & False & False & False \\
 $ \text{FR}_{10}^{\text{9,4}} $ & $ \text{} $ & 9 & 4 & 116 & 4.89898 & 24. & True & False & False & False \\
 $ \text{FR}_{11}^{\text{9,4}} $ & $ [\mathbb{Z}_2 \trianglelefteq \mathbb{Z}_6]_{\mathbf{1}|1}^{( \mathbf{2} \  \mathbf{3} )} $ & 9 & 4 & 117 & 6.63732 & 44.054 & False & False & False & False \\
 $ \text{FR}_{12}^{\text{9,4}} $ & $ \text{} $ & 9 & 4 & 124 & 5.65685 & 32. & True & False & False & False \\
 $ \text{FR}_{13}^{\text{9,4}} $ & $ \text{} $ & 9 & 4 & 132 & 6.15276 & 37.8564 & True & False & False & False \\
 $ \text{FR}_{14}^{\text{9,4}} $ & $ \text{} $ & 9 & 4 & 132 & 6.9282 & 48. & True & False & False & False \\
 $ \text{FR}_{15}^{\text{9,4}} $ & $ \text{} $ & 9 & 4 & 140 & 7.03925 & 49.551 & True & False & True & True \\
 $ \text{FR}_{16}^{\text{9,4}} $ & $ \text{} $ & 9 & 4 & 143 & 7.74597 & 60. & True & False & False & False \\
 $ \text{FR}_{17}^{\text{9,4}} $ & $ \text{} $ & 9 & 4 & 148 & 7.66411 & 58.7386 & False & False & False & False \\
 $ \text{FR}_{18}^{\text{9,4}} $ & $ \text{} $ & 9 & 4 & 148 & 7.66411 & 58.7386 & True & False & False & True \\
 $ \text{FR}_{19}^{\text{9,4}} $ & $ \text{} $ & 9 & 4 & 156 & 8.62364 & 74.3672 & True & False & False & False \\
 $ \text{FR}_{20}^{\text{9,4}} $ & $ \text{} $ & 9 & 4 & 163 & 6.9282 & 48. & False & False & True & False \\
 $ \text{FR}_{21}^{\text{9,4}} $ & $ \text{} $ & 9 & 4 & 163 & 6.9282 & 48. & True & False & True & True \\
 $ \text{FR}_{22}^{\text{9,4}} $ & $ \text{} $ & 9 & 4 & 164 & 9.33766 & 87.1918 & False & False & False & False \\
 $ \text{FR}_{23}^{\text{9,4}} $ & $ \text{} $ & 9 & 4 & 164 & 9.33766 & 87.1918 & True & False & False & False \\
 $ \text{FR}_{24}^{\text{9,4}} $ & $ \text{} $ & 9 & 4 & 165 & 7.23607 & 52.3607 & True & False & False & False \\
 $ \text{FR}_{25}^{\text{9,4}} $ & $ \text{} $ & 9 & 4 & 185 & 10.4077 & 108.321 & True & False & False & False \\
 $ \text{FR}_{26}^{\text{9,4}} $ & $ \text{} $ & 9 & 4 & 195 & 8.78442 & 77.166 & True & False & True & True \\
 $ \text{FR}_{27}^{\text{9,4}} $ & $ \text{} $ & 9 & 4 & 198 & 9.71016 & 94.2873 & False & False & True & False \\
 $ \text{FR}_{28}^{\text{9,4}} $ & $ \text{} $ & 9 & 4 & 200 & 11.6179 & 134.976 & True & False & False & False \\
 $ \text{FR}_{29}^{\text{9,4}} $ & $ \text{} $ & 9 & 4 & 200 & 11.6179 & 134.976 & True & False & False & False \\
 $ \text{FR}_{30}^{\text{9,4}} $ & $ \text{} $ & 9 & 4 & 205 & 9.8961 & 97.9329 & True & False & True & True \\
 $ \text{FR}_{31}^{\text{9,4}} $ & $ \text{} $ & 9 & 4 & 227 & 11.4279 & 130.596 & False & False & False & False \\
 $ \text{FR}_{32}^{\text{9,4}} $ & $ \text{} $ & 9 & 4 & 227 & 11.4279 & 130.596 & True & False & False & False \\
 $ \text{FR}_{33}^{\text{9,4}} $ & $ \text{} $ & 9 & 4 & 319 & 15.0837 & 227.519 & False & False & False & False \\
 $ \text{FR}_1^{\text{9,6}} $ & $ \text{TY}(Q) $ & 9 & 6 & 88 & 4. & 16. & False & False & False & False \\
 $ \text{FR}_2^{\text{9,6}} $ & $ \left.\text{TY(}\mathbb{Z}_8\right) $ & 9 & 6 & 88 & 4. & 16. & True & False & False & False \\
 $ \text{FR}_3^{\text{9,6}} $ & $ [\text{Q} \trianglelefteq \text{Q}]_{\mathbf{1}|1}^{\text{Id}} $ & 9 & 6 & 89 & 4.4014 & 19.3723 & False & False & False & False \\
 $ \text{FR}_4^{\text{9,6}} $ & $ [\mathbb{Z}_8 \trianglelefteq \mathbb{Z}_8]_{\mathbf{1}|1}^{\text{Id}} $ & 9 & 6 & 89 & 4.4014 & 19.3723 & True & False & False & False \\
 $ \text{FR}_5^{\text{9,6}} $ & $ \text{Ising$\times $}\mathbb{Z}_3$ & 9 & 6 & 90 & 3.4641 & 12. & True & False & False & False \\
 $ \text{FR}_6^{\text{9,6}} $ & $ \text{Rep(}D_3\text{)$\times $}\mathbb{Z}_3$ & 9 & 6 & 99 & 4.24264 & 18. & True & False & False & False \\
 $ \text{FR}_7^{\text{9,6}} $ & $ \text{} $ & 9 & 6 & 108 & 5.44702 & 29.67 & True & False & True & True \\
 $ \text{FR}_8^{\text{9,6}} $ & $ \text{} $ & 9 & 6 & 108 & 5.32844 & 28.3923 & True & False & False & False \\
 $ \text{FR}_9^{\text{9,6}} $ & $ [\mathbb{Z}_2 \trianglelefteq \mathbb{Z}_6]_{\mathbf{1}|1}^{\text{Id}} $ & 9 & 6 & 117 & 6.63732 & 44.054 & True & False & False & False \\
 $ \text{FR}_{10}^{\text{9,6}} $ & $ \text{PSU}(2)_5\times \mathbb{Z}_3$ & 9 & 6 & 126 & 5.28088 & 27.8877 & True & True & False & False \\
 $ \text{FR}_{11}^{\text{9,6}} $ & $ \text{} $ & 9 & 6 & 162 & 7.42757 & 55.1689 & True & False & True & False \\
 $ \text{FR}_{12}^{\text{9,6}} $ & $ \text{} $ & 9 & 6 & 198 & 9.71016 & 94.2873 & True & False & True & True \\
 $ \text{FR}_{13}^{\text{9,6}} $ & $ \text{} $ & 9 & 6 & 200 & 11.6179 & 134.976 & False & False & False & False \\
 $ \text{FR}_{14}^{\text{9,6}} $ & $ \text{} $ & 9 & 6 & 200 & 11.6179 & 134.976 & True & False & False & False \\
 $ \text{FR}_{15}^{\text{9,6}} $ & $ \text{} $ & 9 & 6 & 216 & 11.3115 & 127.95 & True & False & False & True \\
 $ \text{FR}_{16}^{\text{9,6}} $ & $ \text{} $ & 9 & 6 & 243 & 12.7817 & 163.373 & True & False & False & False \\
 $ \text{FR}_1^{\text{9,8}} $ & $ \mathbb{Z}_9$ & 9 & 8 & 81 & 3. & 9. & True & True & False & False \\
 $ \text{FR}_2^{\text{9,8}} $ & $ \mathbb{Z}_3\times \mathbb{Z}_3$ & 9 & 8 & 81 & 3. & 9. & True & True & False & False \\
 \hline
\end{longtable}

\end{document}